\documentclass[pdflatex,sn-mathphys-num]{sn-jnl}

\usepackage{graphicx}%
\usepackage{multirow}%
\usepackage{amsmath,amssymb,amsfonts}%
\usepackage{amsthm}%
\usepackage{mathrsfs}%
\usepackage[title]{appendix}%
\usepackage{xcolor}%
\usepackage{textcomp}%
\usepackage{manyfoot}%
\usepackage{booktabs}%
\usepackage{algorithm}%
\usepackage{algorithmicx}%
\usepackage{algpseudocode}%
\usepackage{listings}%
\usepackage[utf8]{inputenc}
\usepackage{natbib}
\usepackage{anyfontsize}
\usepackage{fullpage}
\usepackage{tabularx}
\usepackage{float}
\usepackage{mathtools}
\usepackage{enumitem}
\usepackage{subcaption}
\usepackage{tikz}
\usetikzlibrary{mindmap}

\usepackage[capitalize]{cleveref}

\newtheorem{theorem}{Theorem}
\newtheorem{lemma}[theorem]{Lemma}

\newtheorem{corollary}[theorem]{Corollary}

\theoremstyle{definition}
\newtheorem{definition}[theorem]{Definition}

\newcommand{\bigO}{O} 

\newcommand{\threefield}[3]{$#1\mid#2\mid#3$}

\DeclareMathOperator{\rep}{rep}
\newcommand{\EQUIJIT}{\threefield{1}{\rep}{\min_j \sum_i Z_{i,j}}}
\newcommand{\OPEQUIJIT}{\threefield{1}{\rep, d_{i,j}=d_j}{\min_j \sum_i Z_{i,j}}}
\newcommand{\DDEQUIJIT}{\threefield{1}{\rep, p_{i,j}=p_j}{\min_j \sum_i Z_{i,j}}}
\newcommand{\SEQUIJIT}{\threefield{1}{\rep, p_{i,j}=p}{\min_j \sum_i Z_{i,j}}}
\newcommand{\SDEQUIJIT}{\threefield{1}{\rep, d_{i,j}=d_j, p_{i,j}=p_j}{\min_j \sum_i Z_{i,j}}}
\newcommand{\UEQUIJIT}{\threefield{1}{\rep, p_{i,j}=1}{\min_j \sum_i Z_{i,j}}}

\newcommand{\WEQUIJIT}{\threefield{1}{k_j,\rep}{\min_j \sum_i Z_{i,j}}}

\newcommand{\q}{m} 

\newcommand{\problemdef}[3]{
	\begin{center}\fbox{
	\begin{minipage}{0.95\textwidth}
		\noindent
		#1
		\vspace{5pt}\\
		\setlength{\tabcolsep}{3pt}
		\begin{tabularx}{\textwidth}{@{}lX@{}}
			\textrm{Input:}     & #2 \\
			\textrm{Task:}  & #3
		\end{tabularx}
	\end{minipage}}
	\end{center}
}

\date{}

\begin{document}

\title{Fair Repetitive Interval Scheduling}

\definecolor{lime}{HTML}{A6CE39}
\DeclareRobustCommand{\orcidicon}{%
	\begin{tikzpicture}
	\draw[lime, fill=lime] (0,0) 
	circle [radius=0.16] 
	node[white] {{\fontfamily{qag}\selectfont \tiny ID}};
	\draw[white, fill=white] (-0.0625,0.095) 
	circle [radius=0.007];
	\end{tikzpicture}
	\hspace{-2mm}
}

\foreach \x in {A, ..., Z}{%
	\expandafter\xdef\csname orcid\x\endcsname{\noexpand\href{https://orcid.org/\csname orcidauthor\x\endcsname}{\noexpand\orcidicon}}
}

\newcommand{\orcidauthorA}{0000-0001-8779-0890}
\newcommand{\orcidauthorB}{0000-0002-6379-0383}
\newcommand{\orcidauthorC}{0000-0002-5309-7075}
\newcommand{\orcidauthorD}{0000-0002-4590-798X}
\newcommand{\orcidauthorE}{0000-0002-2709-599X}

\author[1]{\fnm{Klaus} \sur{Heeger}\orcidA{} }\email{heeger@post.bgu.ac.il}

\author[1]{\fnm{Danny} \sur{Hermelin}\orcidB{}}\email{hermelin@bgu.ac.il}

\author[1]{\fnm{Yuval} \sur{Itzhaki}\orcidC{}}\email{ityuval@bgu.ac.il}

\author[2]{\fnm{Hendrik} \sur{Molter}\orcidD{}}\email{molterh@post.bgu.ac.il}

\author[1]{\fnm{Dvir} 
\sur{Shabtay}\orcidE{}}\email{dvirs@bgu.ac.il}

%

\affil[1]{\small Department of Industrial Engineering and Management, Ben-Gurion~University~of~the~Negev, 
Beer-Sheva, 
Israel
}

\affil[2]{\small Department of Computer Science, Ben-Gurion~University~of~the~Negev, 
Beer-Sheva, 
Israel
}


\abstract{
Fair resource allocation is undoubtedly a crucial factor in customer satisfaction in several scheduling scenarios. This is especially apparent in repetitive scheduling models where the same set of clients repeatedly submits jobs on a daily basis. In this paper, we aim to analyze a repetitive scheduling system involving a set of $n$ clients and a set of $m$ days. On every day, each client submits a request to process a job exactly within a specific time interval, which may vary from day to day, modeling the scenario where the scheduling is done Just-In-Time (JIT). The daily schedule is executed on a single machine that can process a single job at a time,  therefore it is not possible to schedule jobs with intersecting time intervals. Accordingly, a feasible solution corresponds to sets of jobs with disjoint time intervals, one set per day. We define the quality of service (QoS) that a client receives as the number of executed jobs over the $m$ days period.
Our objective is to provide a feasible solution where each client has at least $k$ days where his jobs are processed. We prove that this problem is NP-hard even under various natural restrictions such as identical processing times and day-independent due dates. We also provide efficient algorithms for several special cases and analyze the parameterized tractability of the problem with respect to several parameters, providing both parameterized hardness and tractability results.
}









\maketitle

\section{Introduction}

Fairness as an objective is a relatively new and unexplored concept in scheduling problems. In recent years, customer satisfaction has gained prominence as a goal of high importance~\cite{hom2000overview,lin2011critical}.
This shift has naturally led to an ascent in the prioritization of fairness, compared to alternative scheduling objectives \cite{shim2017innovative,
chen2011tradeoffs,
lee1997current,
doulamis2007fair}. Nevertheless, despite the general pursuit of fair resource allocation, there lacks a clear and universally accepted notion of fairness \cite{gupta2020too}. What is deemed fair in one situation typically depends on specific customer preferences, and may not hold true in another. Any set of preferences may necessitate a unique model for fair resource allocation. For this reason, there remains a significant absence of theoretical models that support fair decision-making in service systems, especially in the context of fair scheduling decisions.

Traditionally, deterministic scheduling models focus on a single scheduling horizon (\emph{e.g.}, a single day). Such models typically aim to find an optimal daily schedule from a global perspective, while ignoring fairness issues altogether. When jobs are associated with clients, an optimal daily schedule may result in a large variance in the Quality of Service (QoS) that each client receives. For example, a daily schedule may minimize the total tardiness of all jobs submitted by the clients, at the expense of a single client whose job is significantly tardier than the others. Hence, in such scenarios, solely focusing on the global perspective may lead to poor customer satisfaction and future abandonment of customers.

To promote fairness, one may consider a daily schedule that ensures the highest possible QoS for the \textit{least satisfied} client.
Accordingly, if tardiness is used as a measure of QoS, a fair solution will involve the minimization of the maximum tardiness.
However, in repetitive scheduling problems, where we need to repeatedly provide schedules for the same set of clients, minimizing the maximum tardiness on each day individually may still result in an \textit{unfair} set of schedules, for example when the jobs of a certain client consistently have the maximum tardiness in each day. The goal of this paper is to provide a set of tools for making fair decisions in a specific repetitive scheduling scenario. 

\subsection{Fair Repetitive Scheduling}

We employ a recently introduced framework for studying the \textit{minimum QoS per agent} measure in the context of \emph{fair repetitive scheduling}~\cite{heeger2021equitable,hermelin2021equitable}. In this framework, we aim to serve a set of~$n$ agents over a period of~$m$ days. In the most basic case, each agent requests the execution of a single job on each of the days. Furthermore, we have a single machine to process the jobs on each day, and a single objective function that measures the QoS of each agent. Thus, one might think of the input in this setting as $\q$ instances of a certain single-machine scheduling problem, where each instance contains $n$ jobs. The goal is to determine whether there exist $\q$ schedules, one for each day, such that a minimum level of QoS is guaranteed for every client over the entire $\q$-day scheduling period. This quality threshold is represented by a natural number $k$, which we call the \emph{fairness parameter}.

The work in~\cite{heeger2021equitable} studies the \threefield{1}{\rep}{\max_j {\sum_i U_{i,j}}} problem in the fair repetitive scheduling framework. In this problem, any instance can be thought of as $\q$ instances of the classical \threefield{1}{}{\sum_j U_j} problem, the problem of minimizing the number of tardy jobs on a single machine. The goal in \threefield{1}{\rep}{\max_j {\sum_i U_{i,j}}} is to find $\q$ schedules, one for each day, where no client has his jobs tardy in more than $k$ of the days. In~\cite{hermelin2021equitable}, the setting of~\cite{heeger2021equitable} is extended to a more general framework, and three additional problems are studied: The \threefield{1}{\rep}{\max_j {\sum_i C_{i,j}}} problem where the goal is ensure that the total completion time of the jobs of each client over all $\q$ days is at most $k$, the \threefield{1}{\rep}{\max_j {\sum_i W_{i,j}}} problem where the goal is to ensure that no client has total waiting time that exceeds~$k$, and \threefield{1}{\rep}{\max_j {\sum_i L_{i,j}}}  which uses the total lateness instead of the total waiting time to measure the QoS received by any of the clients. 

\subsection{Fair Repetitive Interval Scheduling} 

All performance measures employed in \cite{heeger2021equitable} and \cite{hermelin2021equitable} are classical and share the characteristic of being \emph{regular}, i.e., non-decreasing functions of the job completion times. Consequently, these metrics encourage the completion of jobs \textit{as early as possible} to minimize unnecessary machine idle times. Nevertheless, modern production management strategies such as JIT and \emph{lean production} emphasize the importance of completing a job \textit{exactly} on its due date, aiming to avoid both earliness and tardiness as much as possible.
In fact, completing a job too early may result in wastage of capital and additional storage costs. Conversely, a belated completion of a job usually entails tardiness penalties. The successful implementation of the JIT paradigm by the Japanese automotive company Toyota is arguably the major factor for the rise of the Japanese motor industry in the global market~\cite{ohno1988toyota,shingo1985revolution}.
Consequently, JIT scheduling has been studied in many different contexts over the last few decades~\cite{
angelelli2014optimal,
baker1990sequencing,
van2015interval,
bouzina1996interval,
hiraishi2002scheduling,
lann1996single,
shabtay2012just,
spieksma1999approximability,
sung2005maximizing}.

Recognizing the significance of JIT and fair scheduling, this paper focuses on analyzing a fair repetitive scheduling problem where the performance measure captures the JIT scheduling concept. Formally, let~$p_{i,j}$ and~~$d_{i,j}$ denote the processing time and due date of the job of client $j \in \{1,\ldots,n\}$ on day $i \in \{1,\ldots,\q\}$. The processing time and due date of a single job then define a \emph{time-interval} $(d_{i,j}-p_{i,j},d_{i,j}]$ during which a job can be processed. A job can be either scheduled within its time interval or not scheduled at all (\emph{i.e.}, rejected).
A single machine that can process only one job at a time is available on each of the $m$ days. Therefore, two jobs cannot be scheduled on the same day if their time-intervals intersect. In such a case we say that these jobs are in \textit{conflict}. A schedule $\sigma_i$ for day  $i \in \{1,\ldots,\q\}$ is simply a subset of the jobs on that day, and it is said to be \emph{feasible} if it contains no conflicting pair of jobs. 
A solution for all $\q$ days is therefore a tuple $(\sigma_1,\ldots,\sigma_\q)$ of $\q$ schedules, one for each day, and it is said to be \emph{feasible} if each $\sigma_i$ is feasible. Given a feasible solution, we let~$Z_{i,j}$ denote a binary performance measure that equals 1 if the job of client $j$ is scheduled on day $i$, and 0 otherwise. The problem we study in this paper, denoted \EQUIJIT, is defined as follows: 

\problemdef{Fair Repetitive Interval Scheduling (\EQUIJIT)}
{Given a set of $n$ clients, where each client $j \in \{1,\ldots, n\}$ have a job with processing time $p_{i,j} \in \mathbb{N}$ and due date~$d_{i,j}\in\mathbb{N}$ for each day $i \in \{1,\ldots,\q\}$, and an integer $k \in \mathbb{N}$.}
{Find $\q$ schedules, $\sigma_1,\ldots,\sigma_\q$, such that $\sigma_i$ is feasible on each day $i \in \{1,\ldots,\q\}$, and $\sum_{i} Z_{i,j} \geq k$ for each $j\in \{1,\ldots,n\}$.}

\subsection{Our Contribution}

We embark on a thorough investigation of the computational complexity of \EQUIJIT, in an attempt to provide a clear picture of when the problem is tractable or not. Our first main result in this regard, presented in \Cref{sec:km}, gives a dichotomy of the NP-hardness of the problem under every possible value of the fairness parameter~$k$: 
\begin{theorem}
\label{thm:sec3}%
The \EQUIJIT\ problem is polynomial-time solvable when $k\in \{0,\q-1,\q\}$, and otherwise it is NP-hard.
\end{theorem}

Following \Cref{thm:sec3}, we proceed to examine \EQUIJIT\ under further restrictions. The first restriction we consider is \textit{day-independent processing times}, the case where the processing time of a job depends solely on its client and is identical for each of the $\q$ days. That is, for each client~$j$ we have $p_{i,j}=p_j$ for all days $i \in \{1,\ldots,\q\}$. We denote this special case by \DDEQUIJIT. As it turns out, the hardness proof of \Cref{thm:sec3} applies for \DDEQUIJIT\ as well. However, when we further restrict the processing times to be \emph{unit}, the problem becomes polynomial-time solvable. Note that the latter special case corresponds to scheduling scenarios with \textit{slots} of fixed time duration for which clients can apply, as commonly used in public health systems~\cite{carpenter2011managing}. 
\begin{theorem}
\label{thm:sec4b}%
The \DDEQUIJIT\ problem is NP-hard. It is polynomial-time solvable when $p_j=1$ for all $j\in\{1,\ldots,n\}$.
\end{theorem}

We next consider the case of \textit{day-independent due dates}, \emph{i.e.}, the case where for each client~$j$ we have $d_{i,j}=d_j$ for all days $i \in \{1,\ldots,\q\}$. This special case, denoted by \OPEQUIJIT, is not only a natural theoretical restriction but also makes sense in practice as often clients have a unique due date that is associated with \emph{e.g.}~a specific delivery time. It turns out that this variant is NP-hard as well. However, we identify two further natural restrictions for which the problem becomes polynomial-time solvable, namely that either the number of days is constant or that processing times are day-independent as well. 
\begin{theorem}
\label{thm:sec4}%
The \OPEQUIJIT\ problem is NP-hard. It is solvable in polynomial time if either of the  following additional restrictions holds:
\begin{itemize}
\item[$(i)$] The number~$\q$ of days is constant.
\item[$(ii)$] The processing times are day-independent. That is, $p_{i,j}=p_j$ for all $i\in\{1,\ldots,\q\}$ and $j\in\{1,\ldots,n\}$.
\end{itemize}
\end{theorem}

In the final part of the paper, we focus on natural parameters of \EQUIJIT\, and study how they affect the complexity of the problem when they are limited in size. The most natural framework for doing this is the theory of parameterized complexity~\cite{downey2012parameterized,cygan2015parameterized}. For the parameters~$k$ and $m$, we know that the problem is NP-hard even when both are constant due to \Cref{thm:sec3}. We show that this is not the case for the parameter $n$, as \EQUIJIT\ admits a fixed-parameter tractable (FPT) algorithm for the number of clients~$n$. We then turn to consider the \textit{treewidth} $\tau$ of the overall conflict graph, the graph with a vertex for each client where two vertices are adjacent if their corresponding clients have conflicting jobs on some day $i \in \{1,\ldots,m\}$ (see \Cref{sec:prelim}). We show that the problem is NP-hard for $\tau=O(1)$, but admits an FPT algorithm for the parameter $\tau +m$. 
\begin{theorem}
\label{thm:sec5}%
The \EQUIJIT\ problem
\begin{itemize}
\item is NP-hard for $\tau=O(1)$, 
\item admits an FPT algorithm with respect to parameter~$\q+\tau$, and
\item admits an FPT algorithm with respect to~$n$.
\end{itemize}
\end{theorem}

\section{Conflict Graphs and Treewidth}
\label{sec:prelim}

The daily conflict relation between the jobs of the clients can be usefully viewed through the prism of graph theory. This relation can be modeled as a \textit{daily conflict graph} which will be used throughout the paper and is defined as follows: 
\begin{definition}\label{def:conflictgraph}
Given an instance $\mathcal{I}$ of \EQUIJIT\ with $n$ clients and $m$ days, the \emph{day $i$ conflict graph} associated with $\mathcal{I}$ is the graph $G_i=(\{1,\ldots,n\},E_i)$ where each vertex $j \in \{1,\ldots,n\}$ is associated with a client $j$ of $\mathcal{I}$. Two vertices $j_1,j_2 \in  \{1,\ldots,n\}$ are adjacent in $G_i$ if their corresponding clients $j_1$ and $j_2$ have a pair of conflicting jobs on day $i$. The \emph{overall conflict graph} associated with~$\mathcal{I}$ is the graph $G = (\{1,\ldots,n\},E_1 \cup \cdots \cup E_m)$.   
\end{definition}
\noindent Note that for each $i \in \{1,\ldots,m\}$, the day $i$ conflict graph is an \emph{interval graph}, an intersection graph of intervals on the real line~\cite{hajos1957uber,golumbic2004algorithmic}. A feasible schedule on day $i$ corresponds to an \textit{independent set} in $G_i$. A partition of the vertices of the graph into independent sets is called a \textit{coloring}, and a graph admits a $\chi$-coloring if its vertices can be partitioned into $\chi$ such independent sets. Both a maximum independent set and a minimum coloring can be computed in $O(n \log n)$ time in an interval graph, given the interval representation of the graph~\cite{gavril1972algorithms}.

We will be interested in \emph{fixed-parameter tractable} algorithms, or FPT algorithms in short, for certain parameters. An FPT algorithm with respect to some parameter~$\kappa$ is an algorithm running in $f(\kappa) \cdot n^{O(1)}$ time, where $f()$ can be any computable function that depends solely on~$\kappa$ and not on~$n$ (see~\cite{downey2012parameterized} for further information on such algorithms). Note that if a problem is hard for $\kappa = O(1)$ then there cannot exist an FPT algorithm unless P=NP. One prominent parameter that we will consider is the \emph{treewidth}~\cite{robertson1986graph} of the overall conflict graph. This parameter is defined through \emph{tree decompositions} as follows: 
\begin{definition}
A \textit{tree-decomposition} of a graph $G = (V,E)$ is a tree $\mathcal{T}= (\mathcal{X},F)$ with a node set~$\mathcal{X} \subseteq \{X\ |\  X \subseteq V\}$ and $F \subseteq \mathcal{X} \times \mathcal{X}$ that upholds the following:
\begin{itemize}
\item $V = \bigcup_{X \in \mathcal{X}} X$. 
 \item For every edge $\{u,v\}\in E$ there exists $X \in \mathcal{X}$ such that $u,v \in X$. 
\item The set of nodes $\mathcal{X}_v = \{ X\ |\ X\in \mathcal{X} \wedge v \in X \}$ containing any vertex $v \in V$ is a connected subgraph in $\mathcal{T}$.
\end{itemize}
The \textit{width} of a tree-decomposition $\mathcal{T}$ is $\max\{ \lvert X\rvert\ |\ X \in \mathcal{X} \}-1$. The \textit{treewidth} $\tau(G)$ of $G$ is the minimum width over all tree decompositions of $G$. 
\end{definition}
\noindent It is known that computing the treewidth of a graph is NP-hard, yet it is FPT with respect to the parameter treewidth~\cite{bodlaender1993tourist}.
Hence, when describing an FPT algorithm with respect to the treewidth parameter, one can assume that a tree decomposition with minimum width is given alongside the input.


\section{The Fairness Parameter}
\label{sec:km}%

In this section, we prove \Cref{thm:sec3} which maps out the effect of parameters~$k$ and~$m$ (the fairness parameter and the number of days) on the computational hardness of \EQUIJIT. We begin with an NP-hardness proof for the case that $k=1$ and $\q=3$ in \Cref{thm:k2m3}, which we later extend to show NP-hardness for every pair $(\q,k)\in \mathbb{N}^2$  such that $0 < k < m-1 $ and $m \geq 3$. We complement these hardness results with \Cref{thm:km-1} where we specify a reduction from \EQUIJIT\ where $k=\q-1$ to \textsc{2-SAT}, essentially showing a polynomial time algorithm for the remaining case.


We begin by showing that \EQUIJIT\ is NP-hard for $k=1$ and $m=3$, even if each job of each client has the same processing time~$p$.

\begin{theorem}\label{thm:k2m3}
The \SEQUIJIT\ problem is NP-hard for $k=1$ and $\q=3$.
\end{theorem}
\begin{proof}

We present a polynomial-time many-one reduction from \textsc{[2,3]–bounded 3–SAT}~\cite{tovey1984simplified} to \SEQUIJIT\ with $k=1$ and $\q=3$.
In \textsc{[2,3]–bounded 3–SAT} we need to decide whether Boolean formula $\phi$ in conjunctive normal form is satisfiable, given the promise that every clause in $\phi$ contains two or three variables and every variable appears in at most three clauses. We can also assume without loss of generality that no variable $x$ appears in all three times either negated or non-negated, that is because we can create an equivalent \textsc{[2,3]–bounded 3–SAT} formula $\phi'$ by setting $x$ to satisfy all three clauses and removing its clauses from $\phi$.

Given a \textsc{[2,3]–bounded 3–SAT} formula $\phi$, we construct an instance $\mathcal{I}$ of \SEQUIJIT\ as follows:

\begin{itemize}
\item For every variable $x$, we create two clients: $x^T$ and $x^F$.
\item For every clause that contains two literals $c_A$, we create two clients: $c^{A_1}$ and $c^{A_2}$.
\item For every clause that contains three literals $c_B$, we create three clients: $c^{B_1}$, $c^{B_2}$, and $c^{B_3}$.
\item We create three ``dummy'' clients: $a_1$, $a_2$, and $a_3$.
\end{itemize}

Let $\alpha$ denote the number of variables in $\phi$, let $\beta_A$ denote the number of two-literals clauses (type~$A$), and let $\beta_B$ denote the number of three-literals clauses (type~$B$). Accordingly, the number of clauses in $\phi$ is~$\beta_A+\beta_B$. 
The first and second day contain 
variable and clause gadgets. The variable gadgets ensure that every variable is consistently set in all of its clauses. The clause gadgets ensure that each clause is satisfied by one of its literals.

\begin{figure}[H]
\centering

\begin{tikzpicture}

  \def\dx{0.05}
  \def\shiftca{4}

  \draw[thick,->] (0-1,0) -- (14.5,0) node[anchor=north west] {};

  \draw[line width=1, |-|]
  ({0.5-0.5}, 2.5) -- ({1.5-0.5}, 2.5) node[anchor=north west] {};
  \node at (({1-0.5}, 3) {$a_1$};
  \draw[line width=1, |-|]
  ({0.5-0.5}, 1.75) -- ({1.5-0.5}, 1.75) node[anchor=north west] {};
  \node at ({1-0.5},2.1) {$a_2$};
  \draw[line width=1, |-|]
  ({0.5-0.5}, 1.) -- ({1.5-0.5}, 1.) node[anchor=north west] {};
  \node at ({1-0.5},0.5) {$a_3$};

\foreach \x\i in
    {
        1/1,2/1,3/\alpha
    }
    {
      \ifthenelse{\x=2}
      {
        \node at ({1+\x+(\x-1)*\dx+0.5}, 1.75) {\dots};
      }
      {
        \draw[line width=1, |-|]
        ({1+\x+(\x-1)*\dx}, 2.) -- ({2+\x+(\x-1)*\dx}, 2.) node[anchor=north west] {};
        \node at ({1+\x+(\x-1)*\dx+0.5}, 2.5) {$x^{T}_{\i}$};
        \draw[line width=1, |-|]
        ({1+\x+(\x-1)*\dx}, 1.5) -- ({2+\x+(\x-1)*\dx}, 1.5) node[anchor=north west] {};
        \node at ({1+\x+(\x-1)*\dx+0.5},1.) {$x^{F}_{\i}$};
      }
    }

    \foreach \x\i in
        {
            1/1,2/1,3/{\beta_A}
        }
        {
          \ifthenelse{\x=2}
          {
            \node at ({\shiftca+1+\x+(\x-1)*\dx+0.5}, 1.75) {\dots};
          }
          {
            \draw[line width=1, |-|]
            ({\shiftca+1+\x+(\x-1)*\dx}, 2) -- ({\shiftca+2+\x+(\x-1)*\dx}, 2) node[anchor=north west] {};
            \node at ({\shiftca+1+\x+(\x-1)*\dx+0.5}, 2.5) {$c^{A_1}_{\i}$};
            \draw[line width=1, |-|]
            ({\shiftca+1+\x+(\x-1)*\dx}, 1.5) -- ({\shiftca+2+\x+(\x-1)*\dx}, 1.5) node[anchor=north west] {};
            \node at ({\shiftca+1+\x+(\x-1)*\dx+0.5},1.) {$c^{A_2}_{\i}$};
          }
        }

        \foreach \x\i in
            {
                1/1,2/1,3/{\beta_B}
            }
            {
              \ifthenelse{\x=2}
              {
                \node at ({2*\shiftca+1+\x+(\x-1)*\dx+0.5}, 1.75) {\dots};
              }
              {
                \draw[line width=1, |-|]
                ({2*\shiftca+1+\x+(\x-1)*\dx}, 2.5) -- ({2*\shiftca+2+\x+(\x-1)*\dx}, 2.5) node[anchor=north west] {};
                \node at ({2*\shiftca+1+\x+(\x-1)*\dx+0.5}, 3.) {$c^{B_1}_{\i}$};
                \draw[line width=1, |-|]
                ({2*\shiftca+1+\x+(\x-1)*\dx}, 1.75) -- ({2*\shiftca+2+\x+(\x-1)*\dx}, 1.75) node[anchor=north west] {};
                \node at ({2*\shiftca+1+\x+(\x-1)*\dx+0.5},2.1) {$c^{B_2}_{\i}$};
                \draw[line width=1, |-|]
                ({2*\shiftca+1+\x+(\x-1)*\dx}, 1.) -- ({2*\shiftca+2+\x+(\x-1)*\dx}, 1.) node[anchor=north west] {};
                \node at ({2*\shiftca+1+\x+(\x-1)*\dx+0.5},0.5) {$c^{B_3}_{\i}$};
              }
            }

\end{tikzpicture}
\label{fig:NP-h-day1}
\caption{Day 1.}
\end{figure}

As all jobs have the same processing time, we can define every job's execution interval solely by its due date.
We set all processing times to 2.
On the first day (see Figure~1), we have the following.
\begin{itemize}
    \item The due date of the jobs of $a_1$, $a_2$, and $a_3$ is $2$.
    
    \item  The due date of the jobs of $x_{\ell}^T$ and $x_{\ell}^F$ is $2\ell+3$.
    
    \item The due date of the jobs of $c_{\ell}^{A_1}$ and $c_{\ell}^{A_2}$ is $2\alpha+2\ell+5$.
    
    \item The due date of the jobs of $c_{\ell}^{B_1}$, $c_{\ell}^{B_2}$, and $c_{\ell}^{B_3}$ is
    $2\alpha+2\beta_A+2\ell+7$.
\end{itemize}

\begin{figure}[H]
\centering
\begin{tikzpicture}

  \def\dx{0.05}
  \def\shiftca{4}

  \draw[thick,->] (0,0) -- (7.5,0) node[anchor=north west] {};

  \node at (({1.75}, 3.5) {$a_1$};
  \node at (({1.75}, 3.1) {$a_2$};
  \node at (({1.75}, 2.7) {$a_3$};
  \draw[line width=1, |-|]
  ({0.5}, 3.5) -- ({1.5}, 3.5) node[anchor=north west] {};
  \draw[line width=1, |-|]
  ({0.5}, 3.1) -- ({1.5}, 3.1) node[anchor=north west] {};
  \draw[line width=1, |-|]
  ({0.5}, 2.7) -- ({1.5}, 2.7) node[anchor=north west] {};

  \node at (({1.8}, 2.2) {$x^T_1$};
  \node at (({1.8}, 1.6) {$x^F_\alpha$};
  \draw[line width=1, |-|]
  ({0.5}, 2.1) -- ({1.5}, 2.1) node[anchor=north west] {};
  \node at (({1}, 1.9) {$\vdots$};
  \draw[line width=1, |-|]
  ({0.5}, 1.5) -- ({1.5}, 1.5) node[anchor=north west] {};

  \node at (({2.}, 1.) {$c^{A_1}_{1}$};
  \node at (({2.}, 0.3) {$c^{A_2}_{\beta_A}$};
  \draw[line width=1, |-|]
  ({0.5}, 0.9) -- ({1.5}, 0.9) node[anchor=north west] {};
  \node at ({1},0.7) {$\vdots$};
  \draw[line width=1, |-|]
  ({0.5}, 0.3) -- ({1.5}, 0.3) node[anchor=north west] {};

  \foreach \x\i in
      {
          1/1,2/1,3/{\beta_B}
      }
      {
        \ifthenelse{\x=2}
        {
          \node at ({0.4*\shiftca+1+\x+(\x-1)*\dx+0.5}, 1.75) {\dots};
        }
        {
          \draw[line width=1, |-|]
          ({0.4*\shiftca+1+\x+(\x-1)*\dx}, 2.5) -- ({0.4*\shiftca+2+\x+(\x-1)*\dx}, 2.5) node[anchor=north west] {};
          \node at ({0.4*\shiftca+1+\x+(\x-1)*\dx+0.5}, 3.) {$c^{B_1}_{\i}$};
          \draw[line width=1, |-|]
          ({0.4*\shiftca+1+\x+(\x-1)*\dx}, 1.75) -- ({0.4*\shiftca+2+\x+(\x-1)*\dx}, 1.75) node[anchor=north west] {};
          \node at ({0.4*\shiftca+1+\x+(\x-1)*\dx+0.5},2.1) {$c^{B_2}_{\i}$};
          \draw[line width=1, |-|]
          ({0.4*\shiftca+1+\x+(\x-1)*\dx}, 1.) -- ({0.4*\shiftca+2+\x+(\x-1)*\dx}, 1.) node[anchor=north west] {};
          \node at ({0.4*\shiftca+1+\x+(\x-1)*\dx+0.5},0.5) {$c^{B_3}_{\i}$};
        }
      }
\end{tikzpicture}
\label{fig:NP-h-day2}
\caption{Day 2.}
\end{figure}

On the second day  (see Figure~2), we have the following.
\begin{itemize}
    \item The due date of the jobs of $a_1$, $a_2$, and $a_3$ is $2$.
    
    \item The due date of the jobs of $x_{\ell}^T$ and $x_{\ell}^F$ is $2$.
    
    \item The due date of the jobs $c_{\ell}^{A_1}$ and $c_{\ell}^{A_2}$ is 2.

    \item The due date of the jobs of $c^{B_1}_\ell$, $c^{B_2}_\ell$, and $c^{B_3}_\ell$ is
    $3+3\ell$.
\end{itemize}

\begin{figure}[H]
\centering

\begin{tikzpicture}

  \def\dx{2}
  \def\shift{0.3}

  \draw[thick,->] (0,0) -- (14.5,0) node[anchor=north west] {};

  \draw[line width=1, |-|]
  ({0.5}, 2.5) -- ({1.5}, 2.5) node[anchor=north west] {};
  \node at (({1}, 3) {$a_1$};
  \draw[line width=1, |-|]
  ({0.5}, 1.75) -- ({1.5}, 1.75) node[anchor=north west] {};
  \node at ({1},2.1) {$a_2$};
  \draw[line width=1, |-|]
  ({0.5}, 1.) -- ({1.5}, 1.) node[anchor=north west] {};
  \node at ({1},0.5) {$a_3$};

\foreach \x\i in
    {
        1/1
        ,2/1,3/\alpha
    }
    {
      \ifthenelse{\x=2}
      {
        \node at ({\shift/2+3+\x+(\x-1)*\dx+0.5+.25}, 1.75) {\dots};
      }
      {

        \draw[line width=1, |-|]
        ({2+\x+(\x-1)*\dx}, 2.) -- ({3+\x+(\x-1)*\dx}, 2.) node[anchor=north west] {};
        \node at ({2.+\x+(\x-1)*\dx+0.5}, 2.5) {$x^{T}_{\i}$};

        \draw[line width=1, |-|]
        ({1.45+\x+(\x-1)*\dx}, 1.5) -- ({2.45+\x+(\x-1)*\dx}, 1.5) node[anchor=north west] {};
        \node at ({1.45+\x+(\x-1)*\dx+0.45},1.) {$c_1(x^{T}_{\i})$};
        \draw[line width=1, |-|]
        ({2.55+\x+(\x-1)*\dx}, 1.5) -- ({3.55+\x+(\x-1)*\dx}, 1.5) node[anchor=north west] {};
        \node at ({2.55+\x+(\x-1)*\dx+0.6},1.) {$c_2(x^{T}_{\i})$};

        \draw[line width=1, |-|]
        ({\shift+4+\x+(\x-1)*\dx+.5}, 2) -- ({\shift+5+\x+(\x-1)*\dx+.5}, 2) node[anchor=north west] {};
        \node at ({\shift+4.+\x+(\x-1)*\dx+0.5+.5},2.5) {$x^{F}_{\i}$};

        \draw[line width=1, |-|]
        ({\shift+3.45+\x+(\x-1)*\dx+.5}, 1.5) -- ({\shift+4.45+\x+(\x-1)*\dx+.5}, 1.5) node[anchor=north west] {};
        \node at ({\shift+3.45+\x+(\x-1)*\dx+0.5+.45},1.) {$c_1(x^{F}_{\i})$};
        \draw[line width=1, |-|]
        ({\shift+4.55+\x+(\x-1)*\dx+.5}, 1.5) -- ({\shift+5.55+\x+(\x-1)*\dx+.5}, 1.5) node[anchor=north west] {};
        \node at ({\shift+4.55+\x+(\x-1)*\dx+0.5+.6},1.) {$c_2(x^{F}_{\i})$};

      }
    }

\end{tikzpicture}

\label{fig:NP-h-day3}
\caption{Day 3.}
\end{figure}

On the third day (see Figure~3), we have the following.
\begin{itemize}
    \item The due date of the jobs of $a_1$, $a_2$, and $a_3$ is $2$.
    
    \item The due date of the job of $x^T_\ell$ is $10\ell-4$. The due date of the job of $x^F_\ell$ is $10\ell+1$.
    
    \item Let $c_1$, $c_2$, and $c_3$ be the clauses in which $x$ appears. As we can assume that $x$ does not homogeneously appear in all three clauses, $x$ must appear negated in at most two clauses. Symmetrically, we can assume that $x$ appears non-negated in at most two clauses. Let $c_1(x^{F}_{\ell})$ and $c_2(x^{F}_{\ell})$ be the literals in which $x$ appears negated, the due date of $c_1(x^{F}_{\ell})$ is $10\ell$.
    If $c_2(x^{F}_{\ell})$ exists (\emph{i.e.}, if $x_\ell$ appears negated in more than one clause), then its due date is $10\ell+2$.
    In the same fashion the due date of $c_1(x^{T}_{\ell})$ is $10\ell-5$ and if $c_2(x^{T}_{\ell})$ exists, then its due date is $10\ell-3$.
\end{itemize}
Finally, we set $k=1$. This finishes the construction of $\mathcal{I}$, which can clearly be computed in polynomial time.
The scheduling period of $\mathcal{I}$ is $3$ days long.

$(\Rightarrow)$
We show that if $\phi$ is satisfiable, then $\mathcal{I}$ has a feasible $1$-fair schedule.
Assume that the assignment satisfies all clauses of $\phi$, then the following schedule is $1$-fair:
The job of client~$a_1$ is to be scheduled on day one, and the jobs of clients $a_2$ and $a_3$ are to be scheduled on days two and three, respectively.
For every $x_\ell$ that is assigned \verb!true!, schedule the job of $x^T_\ell$ on day one. In the same fashion, schedule on day one the job of $x^F_\ell$ for every $x_\ell$ that is assigned \verb!false!. For each variable $x_\ell$ that is set to \verb!true! schedule the job of client $x^{F}_\ell$ on day three, otherwise schedule on day three the job of $x^T_\ell$.
Let the $\ell$-th type $A$ clause be satisfied with the assignment of literal $c^{A_1}$ then on day one schedule the job of $c^{A_2}$ and on day three the job of $c^{A_1}$, otherwise schedule the job of $c^{A_1}$ on day one and the job of $c^{A_2}$ on day three.
From any type~$B$ clause schedule, without loss of generality, assume that $c^{B_1}$ satisfies the clause, then schedule its job on day three. The other clients' jobs $c^{B_2}$ and $c^{B_3}$ are to be scheduled in days one and two respectively.
This schedule is feasible because none of the scheduled jobs conflict, and it is $1$-fair because every client gets at least one job scheduled. 

$(\Leftarrow)$
We show that if 
$\mathcal{I}$
has an $1$-fair schedule, then 
$\phi$ is satisfiable.
Let $\sigma$
be a feasible and $1$-fair schedule and let $x(\sigma)$ be an assignment of $\phi$ that is computed from $\sigma$ in the following fashion: For every $\ell$ we set $x_\ell$ to \verb!true! if $x^T_\ell$ is scheduled in the first day, otherwise we set it to \verb!false!.
Assume that $\phi$ is not satisfiable by $x(\sigma)$, then there must exist a clause $c_\ell$ in $\phi$ that evaluates to~\verb!false!.
Assume that the type $A$ clause $c^A_\ell$ is not satisfied, then both literals $c^{A_1}_\ell$ and $c^{A_2}_\ell$ evaluate to \verb!false!.
Let the literal of $c^{A_1}_\ell$ be satisfied by setting $x_i$ to either \verb!true! or \verb!false!, we denote by $x^Y_i$ the client $x^T_i$ if and only $c^{A_1}_\ell$ is satisfied by setting $x_i$ to \verb!true! (otherwise $x^Y_i$ denotes $x^F_i$).
It is clear that since $\sigma$ is feasible and $1$-fair, one of $a_1$, $a_2$ and $a_3$ must be scheduled in each of the days.
It follows then that because $c^{A_1}_\ell$ and $c^{A_2}_\ell$ conflict on day one and conflict on day two with the jobs of clients $a_1$, $a_2$ and $a_3$, one of them must be scheduled on day three.
Assume, without loss of generality, that $c^{A_1}_\ell$ is scheduled on day three, then the client $x^Y_i$ must be scheduled on day one because it conflicts with $c^{A_1}_\ell$ on day three.
It then follows that $x^Y_i$ satisfies $c_\ell$ because in the assignment that corresponds to $\sigma$, a variable is set to \verb!true! if it is scheduled on day one. 

In the same fashion, assume that the type $B$ clause $c^B_\ell$ is not satisfied, then all three literals $c^{B_1}_\ell$, $c^{B_2}_\ell$ and $c^{B_3}_\ell$ evaluate to \verb!false!.
Assume, without loss of generality, that $c^{B_1}_\ell$ is scheduled in $\sigma$ on day three, then clearly $c^{B_2}_\ell$ and $c^{B_3}_\ell$ are scheduled on days one and two respectively.
Let $c^{B_1}_\ell$ be satisfied by assigning $x_i$ to \verb!true!,  then by the construction $x^T_i$ conflicts $c^{B_1}_\ell$ on day three.
Since $\sigma$ is conflictless it must be that $x^T_i$ is scheduled on day one, yet this contradicts that $c^B_\ell$ is not satisfied by $x(\sigma)$.
Similarly, when $c^{B_1}_\ell$ is satisfied by assigning $x_i$ to \verb!false!, $x^F_i$ must be scheduled on day one which contradicts our assumption that $c^B_\ell$ is not satisfied by $x(\sigma)$.
\end{proof}

We next use the NP-hardness proof for the case of $(\q,k)=(3,1)$ of \Cref{thm:k2m3} as a base case to inductively show NP-hardness for every $(\q,k)\in \mathbb{N}^2$ where $0 < k < \q-1$ and $m\geq 3$. Observe that an NP-hardness result for any pair $(\q,k)$ implies an NP-hardness for the pair $(\q+1,k+1)$, as we can add a single day with no conflicts to any instance $\mathcal{I}$ of \EQUIJIT\ with $\q$ days: Clearly~$\mathcal{I}$ is $k$-fair if and only if it is $(k+1)$-fair when the extra day without conflicts is added. Moreover, observe that NP-hardness for any pair $(\q,1)$ implies NP-hardness for $(\q+1,1)$, as can be seen by adding an additional client and an additional day to an instance $\mathcal{I}$ such that the new client conflicts every other client on each of first $\q$ days, and on the new day the jobs of \textit{all} clients are mutually conflicting. If $\mathcal{I}$ is $1$-fair, then we can clearly schedule a job of the new client on day $\q+1$ to obtain a $1$-fair schedule for the new instance. Conversely, if the new instance is $1$-fair, then we can assume by a swapping argument that the new client has its job scheduled on day $\q+1$, implying that $\mathcal{I}$ is $1$-fair as well. Observe that this can be done with day-independent processing times, i.e., each client has the same processing time throughout all days. Using these two observations, we can reduce the case of~$(3,1)$ to every instance $(3+i+j,1+i)$ by using $i$ times the first reduction after applying $j$ times the second reduction. 

\begin{corollary}
\label{cor:NPhard-pij-pj}%
The \DDEQUIJIT\ problem is NP-hard when $0<k < m-1$ and $m\geq 3$.
\end{corollary}

Returning to the dichotomy of NP-hardness of \EQUIJIT\ with respect to the parameters $k$ and $\q$, it is clear that the cases where $k=0$ and $k \geq \q$ are straightforward. When $k=0$ every feasible schedule is $k-$fair, making it trivially a YES instance.
When $k=\q$, only conflict-free instances are classified as YES instances, as every job of every client must be scheduled (instances where $k>\q$ are as well trivially NO instances).
These observations leave the remaining case of \EQUIJIT\ when $k=\q-1$, for which we give a polynomial time algorithm using a reduction to \textsc{2-SAT}.

\begin{theorem}
\label{thm:km-1}
The \EQUIJIT\ problem is solvable in $\bigO(\q n^2+n\q^2)$ time when $k=\q-1$.
\end{theorem}
\begin{proof}
To show the above theorem we reduce \EQUIJIT\ with $k=\q-1$ to \textsc{2-SAT}.
With each combination of client $j$ and day $i$ we associate a variable $x_{i,j}$. We then define two sets of clauses. The first set (denoted as the set of \emph{conflicting clauses}) ensures that no two conflicting jobs are scheduled, while the second set (denoted as the set of \emph{validation clauses}) ensures that for each client not more than one job is discarded. Accordingly, for every pair of clients $j_1$ and $j_2$ that are conflicting on day $i$, we add the conflict clause $(\neg x_{i,j_1} \vee \neg  x_{i,j_2})$ that is unsatisfied if both variables are set to \verb!true!. Moreover, for every pair of client's jobs in any two different days, $i_1$ and $i_2$ ($1 \leq i_1 < i_2 \leq \q$) we add a validation clause $(x_{i_1,j} \vee x_{i_2,j})$ which cannot be satisfied if both variables are set to $\verb!false!$.

The size of this construction is polynomial in the size of the input with $\bigO(n\q)$ variables. We create
$\bigO(n\q^2)$ clauses to ensure that at least $\q-1$ jobs of every client are scheduled with $\binom{\q}{2}$ many such clauses per client. For every day we keep additional $\bigO(n^2)$ clauses, as many as the edges in the conflict graph, which ensures that no two conflicting jobs are simultaneously scheduled. Hence there are $\bigO(\q n^2+n\q^2)$ clauses in the constructed \textsc{2-SAT} formula which can be solved in linear time with respect to its size~\cite{aspvall1979linear}.

$(\Leftarrow)$
If there exists a $k$-fair schedule for \EQUIJIT, then the \textsc{2-SAT} is satisfiable.
To obtain a satisfying assignment from a feasible schedule, simply set $x_{i,j}$ to $\verb!true!$ if and only if the job of client $j$ is scheduled on day $i$.
In a feasible schedule, no two scheduled jobs are in conflict.
It follows that in no conflicting clause both variables are set to true, which means that all conflicting clauses are satisfied by a feasible schedule.
If the schedule is $k$-fair (for $k=\q-1$), then no client has two or more rejected jobs.
All validation clauses are therefore satisfied as every validation clause has at least one variable set to $\verb!true!$.

$(\Rightarrow)$
If the \textsc{2-SAT} formula is satisfiable, then there exists a $k$-fair schedule. To obtain such a schedule, we schedule client $j$'s job on day $i$ if and only if $x_{i,j}$ is $\verb!true!$ in the assignment.
We know that none of these jobs are in conflict as the conflict clauses are satisfied and hence the schedule is feasible. We know as well that for every client at most one job is rejected, otherwise a validation clause will not be satisfied.
\end{proof} 

Combining \Cref{cor:NPhard-pij-pj} with \Cref{thm:km-1} completes the proof of \Cref{thm:sec3}. Observe that \Cref{cor:NPhard-pij-pj} also gives us the first part of \Cref{thm:sec4b}.


\section{Day-Independent Due Dates and Processing Times}
\label{sec:op}%

In this section, we will examine the \EQUIJIT\ problem in the case where either processing times are day-independent, or where the due dates are day-independent, \emph{i.e.}, problems \DDEQUIJIT\ and \OPEQUIJIT. We begin by completing the proof of \Cref{thm:sec4b} by showing that \UEQUIJIT, the special case of \EQUIJIT\ where all processing times are of unit length, is polynomial-time solvable. This is done via a reduction to the \textsc{Bipartite Maximum Matching} problem.

\begin{theorem}
\label{thm:bipartite}%
The \UEQUIJIT\ problem is solvable in $\bigO(n^{1.5} \cdot \q^{2.5})$ time.
\end{theorem}

\begin{proof}
We show the above theorem with a reduction from \UEQUIJIT\:to the \textsc{Bipartite Maximum Matching} problem. 
Informally, we create a bipartite graph with the set of \textit{job vertices} on one side and the set of \textit{due dates vertices} on the other, to the latter we add a set of \textit{rejection vertices}, $\q-k$ per client. A matching in this graph where all job vertices are matched corresponds to a schedule in which every job is either completed at the due date to which it was matched or is rejected. The schedule is fair if \textit{all} job vertices are matched since the number of rejection vertices of every client is bounded by $\q-k$ (see \Cref{fig:bipartiteM} for illustration).

Given an instance of \UEQUIJIT, we create the following instance of the \textsc{Bipartite Maximum Matching} problem:
\begin{itemize}
\item For each $j\in \{1,\ldots,n\}$ and each $i\in \{1,\ldots,\q\}$, we create a \textit{job} vertex $v_{i,j}$. The set of vertices $V = \{v_{i,j} : 1 \leq i \leq n, 1 \leq j \leq \q\}$ represents all input jobs of all clients over all days.
\item For each $i\in \{1,\ldots,\q\}$ and each $d\in \{d_{i,j} : 1 \leq j \leq n\}$, we create a \textit{due date} vertex $u_{i,d}$ if no such vertex already exists. The set $U=\{u_{i,d_{i,j}}: 1 \leq j \leq n, 1\leq i \leq \q\}$ represents all distinct possible due dates of all input jobs that get scheduled.
\item For each $j \in \{1,\ldots,n\}$ and each $\ell \in \{1,\ldots,\q-k\}$, we create a \textit{rejection} vertex $w_{\ell,j}$. Matching a job vertex to a vertex of $W=\{ w_{\ell,j} : 1\leq j \leq n, 1\leq \ell \leq \q-k \}$ is analogous to rejecting the corresponding job.
\end{itemize}

The edges of $G$ are constructed as follows. For each $j\in \{1,\ldots,n\}$ and each $i\in \{1,\ldots,\q\}$ we connect $v_{i,j}$ to:
\begin{itemize}
\item vertices $w_{1,j},\ldots, w_{\q-k,j}$, and
\item vertex $u_{i,d_{i,j}}$.
\end{itemize}

For the correctness, we prove that $G$ has a matching of size $n\q$ if and only if there exist a schedule $\sigma$ where every client is served on at least $k$ days.

$(\Rightarrow)$
We show how to obtain a $nm$ size matching $M$ given a $k$-fair schedule $\sigma$. Let $\sigma$ be a schedule where every client is served on at least $k$ days. Consider the job of client $j$ on day $i$, for some $j \in \{1,\ldots,n\}$ and $i \in \{1,\ldots,\q\}$. If the job is scheduled in $\sigma_i$ and its completion time is $d=d_{i,j}$, then there is an edge
$\{v_{i,j},u_{i,d}\}$ in $G$. We add this edge to the matching $M$. Otherwise the job of client $j$ is not scheduled on day $i$. Let $\ell$ denote the number of days prior to $i$ in which the job of client $j$ is not scheduled. Then $\ell < m-k$, as otherwise $\sigma$ is not a $k$-fair schedule. We add the edge $\{v_{i,j},w_{\ell+1,j}\}$ to the matching $M$. Observe that no pair of edges in $M$ have a common vertex, and hence $M$ is matching of size $nm$ in $G$.

$(\Leftarrow)$
Assume that $G$ contains a matching $M$ of size~$nm$. First note that the fact that $G$ is bipartite with one part being $V$, and $|V|=\q n$, implies that every vertex in the set of job vertices $V$ has to be matched. Let $v_{i,j}\in V$. Then this vertex is either matched to a vertex in $U$ or a vertex in $W$. We create the schedule $\sigma$ as follows.
\begin{itemize}
\item Suppose that $v_{i,j}$ is matched to some vertex $w_{\ell_0,j_0}\in W$. Note that $j_0 = j$ and $\ell_0 \leq \q-k$ by the construction of $G$. Then we do not schedule the job of client $j$ on day $i$.
\item Suppose that $v_{i,j}$ is matched to some $u_{i_0,d}\in U$. Observe that $i=i_0$ and $d = d_{i,j}$ by the construction of $G$. We set $\sigma$ so that the job of client $j$ is scheduled on day $i$. Note that the fact that $u_{i_0.d}$ cannot be matched to any other vertex in $V$ guarantees that no two jobs can be simultaneously scheduled.
\end{itemize}

The schedule $\sigma$ is feasible with no conflicting jobs, otherwise, there are two job vertices that are matched to the same due date vertex. Assume towards a contradiction that $\sigma$ is not a $k$-fair schedule, then there is some client with less than $k$ jobs scheduled over the period of $m$ days. However, this leads to a contradiction as every vertex in $V$ must be matched, either to a vertex in $U$ or $W$. If client $j$ has less than $k$ scheduled jobs, then less than $k$ vertices of type $v_{i,j}$ are matched to vertices in $U$ for $i\in \{1,\dots,\q\}$, and the rest of these vertices can only be matched to at most $m-k$ vertices of type $w_{i,j}$ in $W$.

We observe that $G$ has $O(n\q)$ vertices and $O(\q^2n)$ edges, and it can be constructed in $O(\q^2n)$ time. Using the algorithm of \citet{hopcroft1973n} for the \textsc{Bipartite Maximum Matching} problem we can solve \UEQUIJIT~in $O(E\sqrt{V})$ where $E$ and $V$ are the number of edges and vertices in the graph. The theorem thus follows.
\end{proof}

\begin{figure}
    \centering
\begin{tikzpicture}

  \node[] at (1.2*1+3.5+1.75-4.5, -1.) {$\q$};
   \draw[line width=1, |-|]
  (1.2+3.5-4.5, -.7) -- (1.2*2+5.75-4.5, -.7) node[anchor=north] {};
  \foreach \i in {1,...,3} {
    \node[draw, circle, fill=gray!40] (A\i) at (\i*1.2-.5, 0) {$v_{\i,j'}$};
  }

  \node[] at (1.2*1+3.5+1.75, -1.) {$\q$};
   \draw[line width=1, |-|]
  (1.2+3.5, -.7) -- (1.2*2+5.75, -.7) node[anchor=north] {};
 
  \foreach \i in {1,...,3} {
    \node[draw, circle, fill=gray!40] (C\i) at (\i*1.2+4, 0) {$v_{\i,j}$};
  }
  
  \node[] at (1.2*1+4.6, 5.1) {$\q-k$};
   \draw[line width=1, |-|]
  (1.2+3.5, 4.75) -- (1.2*2+4.5, 4.75) node[anchor=north] {};
  \foreach \i in {1,...,2} {
    \node[draw, circle, fill=gray!60] (B\i) at (\i*1.2+4, 4) {$w_{\i,j'}$};
  }

  \node[] at (1.2*1+7.6, 5.1) {$\q-k$};
   \draw[line width=1, |-|]
  (1.2+6.5, 4.75) -- (1.2*2+7.5, 4.75) node[anchor=north] {};
  \foreach \i in {1,...,2} {
    \node[draw, circle, fill=gray!60] (D\i) at (\i*1.2+7, 4) {$w_{\i,j}$};
  }

    \node[draw, circle, fill=gray!10] (F1) at (1*1.2-2, 4) {$u_{1,d_1}$};
    \node[draw, circle, fill=gray!10] (F2) at (2*1.2-2, 4) {$u_{1,d_2}$};
    \node[draw, circle, fill=gray!10] (F3) at (3*1.2-2, 4) {$u_{2,d_3}$};
    \node[draw, circle, fill=gray!10] (F4) at (4*1.2-2, 4) {$u_{3,d_4}$};

  \foreach \i in {1,...,3} {
    \foreach \j in {1,...,2} {
      \draw (A\i) -- (B\j);
      \draw (C\i) -- (D\j);
    }
  }

    \draw (A1) -- (F1);
    \draw (C1) -- (F2);

    \draw (C2) -- (F3);
    \draw (A2) -- (F3);

    \draw (C3) -- (F4);
    \draw (A3) -- (F4);
    
\end{tikzpicture}

\caption{
The construction of a bipartite graph from an instance of \UEQUIJIT\ with $\q=3$ and $k=1$, where only the relevant vertices for two clients $j$ and~$j'$ are depicted. Clients $j$ and~$j'$ have different due date on day 1, and the same due date on day 2 and 3. Both jobs of client~$j$ and~$j'$ have two rejection vertices with which they can also be potentially matched.}
\label{fig:bipartiteM}%
\end{figure}
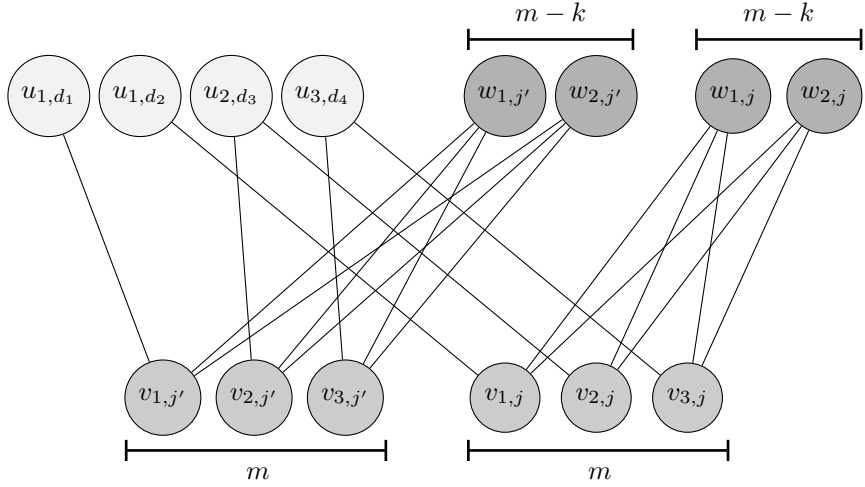

Note that \Cref{thm:bipartite} together with \Cref{cor:NPhard-pij-pj} provides the proof of \Cref{thm:sec4b}. We next turn our attention to day-independent due dates, \emph{i.e.}, the case where $d_{i,j}=d_j$. As it turns out, this variant when $k=1$ is rather similar to the  \threefield{R}{}{\sum_j Z_j} problem. In \threefield{R}{}{\sum_j Z_j} we are given a set of $n$ jobs, where job $j$ has a due date $d_j$. We are given as well a set of $m$ machines with unrelated speeds, \emph{i.e.}, $p_{i,j}$ is the processing time of job $j$ on machine $i$. We are then asked whether we can execute \textit{all} jobs in a JIT mode; that is, whether we can schedule all jobs on the set of $m$ machines such that all jobs on each machine are non-conflicting. The \threefield{R}{}{\sum_j Z_j} is strongly NP-hard~\cite{sung2005maximizing} and was recently proven to be W[1]-hard with respect to the number of machines~\cite{hermelin2022hardness}.

We first show that the NP-hardness of \threefield{R}{}{\sum_j Z_j} implies that \OPEQUIJIT\ is NP-hard.

\begin{theorem}\label{thm:op-np-w1-hard}
The \OPEQUIJIT\ problem is NP-hard.
\end{theorem}
\begin{proof}
We give a reduction from \threefield{R}{}{\sum_j Z_j} to \OPEQUIJIT\ with $k=1$. Given an instance of \threefield{R}{}{\sum_j Z_j} we create an instance of \OPEQUIJIT\ as follows: For every job $j$ we create a client $c_j$, and for every machine $i$ create a day $i$. The due dates are day-independent. Accordingly, $d_j$ is the common due-date of all jobs belonging to client $c_j$. The processing time of client $c_j$'s job on day $i$ is $p_{i,j}$ (that is, the processing time of job $j$ on machine $i$). We then ask if there exists a $1$-fair schedule for the constructed instance of \OPEQUIJIT.~We show below that the answer for this question is yes if and only if all jobs of \threefield{R}{}{\sum_j Z_j} can be scheduled in a JIT mode.

$(\Leftarrow)$
Given a $1$-fair schedule for \OPEQUIJIT\ we obtain a schedule $\sigma$ for \threefield{R}{}{\sum_j Z_j} by scheduling on machine $i$ all jobs that are scheduled on day $i$ of \OPEQUIJIT.       Clearly, all jobs are scheduled in $\sigma$ as each of the $n$ clients was served at once in the $1$-fair schedule.
There are no conflicts on any machine as there were no conflicts on any of the days in the $1$-fair schedule.
    
$(\Rightarrow)$
Given a schedule $\sigma$ for \threefield{R}{}{\sum_j Z_j} we obtain a $1$-fair schedule for \OPEQUIJIT\ by scheduling on day $i$ all jobs that are machine $i$ of \threefield{R}{}{\sum_j Z_j}.
Clearly, we have created a conflictless and $1$-fair as all jobs were conflictlessly scheduled in $\sigma$, hence every client is satisfied on exactly one day.
\end{proof}

The problem \threefield{R}{}{\sum_j Z_j} admits a
polynomial time algorithm for a constant number of machines~\cite{sung2005maximizing}. Note that this result carries over to \OPEQUIJIT\ where $k=1$, as can be shown by the inverting the construction of \Cref{thm:op-np-w1-hard}. Nevertheless, our next result generalizes this algorithm for any value of $k$ and proves that \OPEQUIJIT\ is polynomial-time solvable when the number of days $\q$ is constant. Recall that $k\le m$ in all non-trivial instances, hence we can assume that if $m$ is constant, then $k$ is constant as well.

\begin{theorem}
\label{thm:opxp}%
The \OPEQUIJIT\ problem is solvable in $O(m^{k+1}n^{m+1})$ time.
\end{theorem}

\begin{proof}
Given an instance of \OPEQUIJIT\, we order the set of clients $\{1,\ldots,n\}$ in $O(n \log n)$ time in non-decreasing order of their due dates. Thus, assume henceforth that we have~$d_1 \leq d_2 \leq \cdots \leq d_n$. Let $j^{*} \in \{1,\ldots,n\}$, and consider a partial schedule $\sigma=(\sigma_1,\ldots,\sigma_m)$ defined on the of clients $\{1,\ldots,j^{*}\}$ (that is, $\sigma_i \subseteq \{1,\ldots,j^{*}\}$ for each $i \in \{1,\ldots,m\}$). We say that $\sigma$ is feasible if each $\sigma_i$ contains no conflicting jobs.
Further, we let $\Sigma_{j^{*}}$ denote the set of all feasible schedules on $\{1,\ldots,j^{*}\}$ with $\sum_i Z_{i,j} = k$ for every $j \in \{1, \ldots, j^*\}$ (that is, \emph{exactly} $k$ of the jobs of each client from $\{1, \ldots, j^*\}$ are scheduled). An instance of \OPEQUIJIT\ is clearly a YES instance if and only if $\Sigma_n \neq  \emptyset$. 

We represent a schedule $\sigma \in \Sigma_{j^{*}}$ with the state $[j^{*},j_1,\ldots,j_\q]$ in which $j_i$ represents the index of the last scheduled client on day~$i$ ($i \in \{1,\ldots,n\}$). Consider $\Sigma_1$.
It is obvious that every schedule~$\sigma$ for the client set $\{1\}$ that schedules $k$ jobs of client 1 is feasible and thus contained in $\Sigma_1$.
Thus, we have~$|\Sigma_1| = \binom{m}{k}$, and~$\Sigma_1$ can be computed in $O(m^k)$ time. We next show how to compute the set $\Sigma_{j^{*}}$ from $\Sigma_{j^{*}-1}$. 

To compute $\Sigma_{j^{*}}$ from $\Sigma_{j^{*}-1}$ we iterate over the $\binom{\q}{k}$ combinations for scheduling client $j^{*}$'s jobs. Let~$S$ be one of the aforementioned combinations with $S \subseteq \{1,\ldots,\q\}$ and $|S|=k$.
For every schedule $[j^{*}-1, j_1,\ldots,j_m] \in \Sigma_{j^{*}-1}$ we check whether $p_{i,j^{*}}\leq d_{j_0} - d_{j_i}$ for every $i\in S$, \emph{i.e.}, if the jobs of client $j^{*}$ can be feasibly scheduled in the partial schedule $[j^{*}-1, j_1,\ldots,j_\q]$. If yes, we add $[j^{*},j^{'}_1,\ldots,j^{'}_m]$ to~$\Sigma_{j^{*}}$, where $j^{'}_i=j_0$ if $i\in S$ and $j^{'}_i=j_i$ otherwise.

The correctness of the above algorithm is clear. To analyze its time complexity, note that for any $j^{*}\in\{1,\ldots,n\}$ there are $O(n^m)$ possible schedules of the form $[j^{*}, j_1,\ldots,j_m]$, and so the size of each $\Sigma_{j^{*}}$ is $O(n^m)$. We compute all schedules in $\Sigma_{j^{*}}$ by iterating through all $O(m^k)$ combinations~$S$, and for each such $S$ and every $[j^{*}-1, j_1,\ldots,j_m] \in \Sigma_{j^{*}-1}$ we require $O(k)$ time, which is naturally bounded by $O(m)$. Thus, in total, we compute $\Sigma_{j^{*}}$ in $O(m^{k+1}n^m)$ time. As there are $n$ such sets to compute, we get the stated running time of the theorem.
\end{proof}

As a final result for the section, we consider the case where both the processing times and due dates are day-independent, \emph{i.e.}, the \SDEQUIJIT\ problem. This corresponds to the case that each client has the same time interval for each day, and so the conflict graph is identical for each of the days.
Recall that daily conflict graph is an interval graph,
one can compute the chromatic number $\chi(G)$ of any interval graph $G$ with $n$ vertices in $O(n\log n)$ time given its interval representation~\cite{gavril1972algorithms}. In scheduling terms, the chromatic number is the minimum number of machines required to schedule all jobs in a JIT mode. We use this fact to show that \SDEQUIJIT\ can be solved in $O(n\log n)$ time.

\begin{theorem}\label{thm:chromatic}
The \SDEQUIJIT\ problem is solvable in $O(n\log n)$ time. 
\end{theorem}
\begin{proof}
Let $G=(V,E)$ be the conflict graph of a given \SDEQUIJIT\ instance. As mentioned above, we can compute the chromatic number $\chi=\chi(G)$ of $G$ in $O(n\log n)$ time. We show that there exists a feasible $k$-fair schedule for our given instance if and only if~$k \cdot \chi \leq\q$.

If $k\cdot \chi \leq \q $, then there exist $\chi$ independent sets in $G$ whose union is $V$. We can therefore schedule each of the $\chi$ independent sets on a different day. Since there are at least $k\cdot \chi$ days, each of the $\chi$ independent sets can be scheduled $k$ times without conflicts. If $k \cdot \chi > \q $, then there cannot exist a $k$-fair schedule due to the fact that interval graphs are perfect graphs, meaning that their chromatic number equals the size of their largest clique~\cite{golumbic2004algorithmic}. In our case, a clique is a subset of jobs that are mutually conflict with each other. Clearly, each of the members of the largest clique must be scheduled on a separate day and must be satisfied at least $k$ times. Seeing that we can only schedule one client in the clique per day, if $k\cdot \chi > \q $ there are simply not enough days to schedule all clients of the clique.
\end{proof}

\Cref{thm:sec4} now follows directly from \Cref{thm:op-np-w1-hard}, \Cref{thm:opxp}, and \Cref{thm:chromatic}.

\section{Complexity of Strucural Parameters}\label{sec:fpt}

This section focuses on FPT algorithms for \EQUIJIT; namely, we provide a complete proof of \Cref{thm:sec5}. As we have seen from \Cref{thm:k2m3}, the problem is also NP-hard for a constant number of days $\q$, and so we begin by showing an NP-hardness for instances where the overall conflict graph (\Cref{def:conflictgraph}) has a constant treewidth $\tau$. We then present an FPT algorithm for parameter~$m+\tau$, and finish the section by presenting an FPT algorithm with respect to~$n$.

\subsection{Hardness with respect to treewidth}
\label{sec:paraw}%

We begin by showing that \EQUIJIT\ is NP-hard even when the treewidth~$\tau$ of the overall conflict graph is constant. For this, we first show a hardness result for a generalization of {\EQUIJIT}, and then we reduce this generalized version to {\EQUIJIT}\ such that the treewidth increases only by a small constant. In the generalized version \WEQUIJIT\ of \EQUIJIT, we are given a number~$k_j$ for each client~$j$, and are asked to find a schedule with $\sum_i X_{i,j} \ge k_j$ for every $j$. Thus, each client has its own fairness parameter, instead of the instance having a single global parameter.

\begin{lemma}\label{lem:td-weighted}
The \WEQUIJIT\ problem is NP-hard even when the overall conflict graph has treewidth at most 4. 
\end{lemma}

\begin{proof}
We reduce from the NP-hard \textsc{Multicolored Independent Set} problem~\cite{Pietrzak03,fellows2009parameterized}. In \textsc{Multicolored Independent Set}, the input is an $\ell$-partite graph $G=(V_1 \cup\ldots\cup V_\ell, E)$, where each vertex class in the partition is referred to as a \emph{color class} (and vertices $v \in V_i$ are of \emph{color $i$}).  The question is whether there exists a multicolored independent set $I$ of size $\ell$ in $G$, \emph{i.e.}, a set $I$ of nonadjacent vertices in $G$ with $|I \cap V_i| =1$ for each $i \in \{1,\ldots,\ell\}$. Let $G = (V_1 \cup \ldots \cup V_\ell, E)$ be an instance of \textsc{Multicolored Independent Set}. We can assume without loss of generality that~$G$ is $r$-regular (\emph{i.e.}, every vertex has $r$ neighbours), that $|V_1| = \cdots = |V_\ell| = n$, and that $|E|$ is even.
Let $V_i = \{v_1^i, \ldots, v_n^i\}$ for each~$i \in \{1,\dots,\ell\}$ and let $E = \{e_1, \ldots, e_{|E|}\}$.
To simplify notation, we assume that $G$ is directed, where each edge connecting $v \in V_i$ to $u \in V_j$ with $i < j$ is directed from~$V_i$ to~$V_j$.

We next describe the construction of the \WEQUIJIT\ instance. The instance will contain $\ell \cdot (n+1) + |E|$ days:
For every color $i\in \{1,\dots,\ell\}$, we create $n$ \textit{vertex days} and a single \textit{validation day}. These will allow us to encode which vertex from $V_i$ is selected to be in the independent set. For every edge $e\in E$ we create a single \textit{edge day}, which will help in ensuring that the selection done in vertex days indeed encodes a multicolored independent set in $G$. The clients of the \WEQUIJIT\ instance are constructed as follows:
\begin{itemize}
\item For vertex $v \in V_1 \cup \cdots \cup V_\ell$ we create a \textit{vertex client} $c_v$ with a fairness parameter of ~$k_v=1$.
\item For each color $i\in\{1,\dots,\ell\}$ we create a \textit{selection client} $c_i$ with fairness parameter $k_i=1$.
\item For each edge $(v,u)\in E$ we create two edge clients $c_{v,u}$ and $c_{u,v}$ with $k_{v,u}=k_{u,v}=1$.
\item We create additional two \textit{interaction clients} $c^+$ and $c^-$ with $k^+=k^-=|E|/2$.
\item We create a single \textit{dummy client} $c_0$ with $k_0 = \ell\cdot(n+1)+|E|$.
\end{itemize}

The general idea of the construction is that for every color the \textit{vertex days} and \textit{validation day} form a \textit{vertex selection gadget} which compels the selection of a single \textit{vertex client} for the color (\Cref{fig:vsg}).
The \textit{edge days} (\Cref{fig:edgeday}) form an \textit{incidence-checking gadget} which certifies that all clients are satisfied only if the selected vertices form a multicolored independent set in $G$. More specifically, for every vertex $v \in V_1 \cup \cdots \cup V_\ell$ we create $v$'s \textit{vertex day} as follows:
\begin{itemize}
\item Client $c_{v}$'s job has due date $r+1$ and processing time $r$.
\item Client $c_{i}$'s job has due date $r+1$ and processing time $r$.
\item All other jobs have due date $1$ and processing time $1$.
\end{itemize}
\noindent For every color $i \in \{1,\ldots,\ell\}$ we create $i$'s \textit{validation day} as follows:
\begin{itemize}
\item For every $v=v^i_p \in   V_i$ we set the due date of $c_v$'s job to $r(p+1)$ and its processing time to $r$.
\item For every $v=v^i_p  \in  V_i$ we define an arbitrary order over $v$'s neighbors. Let $u$ be the $q$'th neighbor of $v$, then $c_{v,u}$ has due date $rp + q$ and processing time $1$.
\item All other jobs have due date $1$ and processing time $1$.
\end{itemize}
\noindent For every edge $(v,u) \in E$ we create $(v,u)$'s \textit{edge day}:
\begin{itemize}
\item Client $c^-$'s job has due date $3$ and processing time $2$.
\item Client $c^+$'s job has due date $4$ and processing time $2$.
\item Client $c_{v,u}$'s job has due date $4$ and processing time $1$.
\item Client $c_{u,v}$'s job has due date $2$ and processing time $1$.  
\item All other jobs have due date $1$ and processing time $1$.
\end{itemize}

Observe that the jobs of the dummy client must be scheduled on all $\ell (n+1)+|E|$ days to satisfy its fairness requirement. This is used in the gadget creation to \textit{``block''} the execution of jobs that are irrelevant for the gadget. In this way, every vertex client $c_v$ can be only scheduled on $v$'s \textit{vertex day} or on the \textit{validation day} that corresponds to $v$'s color.
For every color $i\in \{1,\dots,\ell\}$, the selection client $c_i$ can be only scheduled on one of $i$'s \textit{vertex days}, as it is blocked by the dummy client in the rest of the days. This means that $\ell$ vertex clients must be scheduled on the $\ell$ \textit{validation days}, as the $\ell$ selection clients block them on their vertex days. See \Cref{fig:vsg} for illustration.

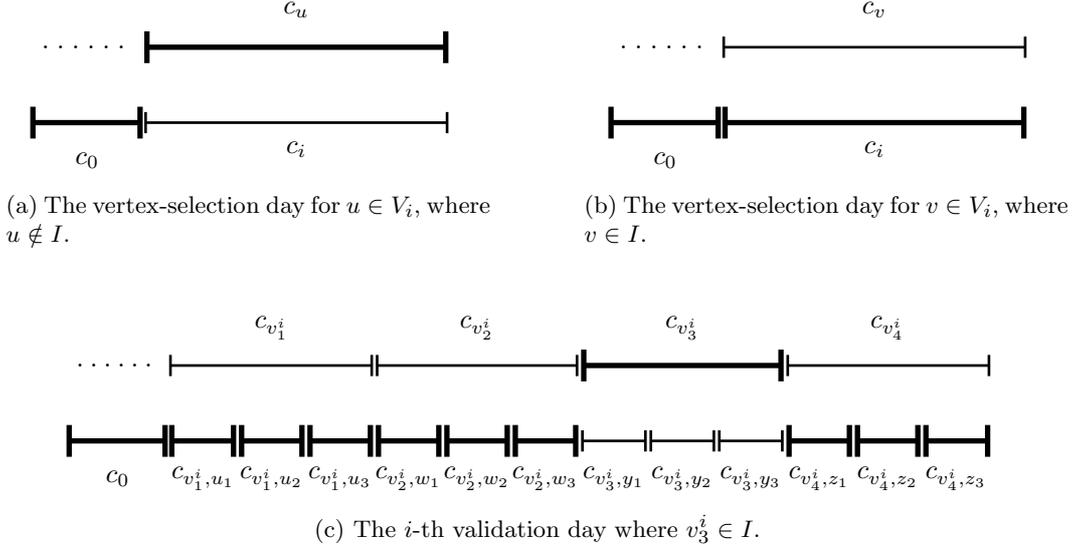
\begin{figure}[H]
\centering
\begin{subfigure}{0.4\textwidth}
\begin{center}
\begin{tikzpicture}
\draw[line width=2, |-|]
  ({-1}, 1.5) -- ({.48}, 1.5) node[anchor=north west] {};
  \node at (({-.25}, 1.) {$c_0$};
  \node at ({-0.25}, 3.) {};
  \foreach \x in
    {
        1,2,3,4,5,6
    }
    {
        \node at ({-1+\x*0.2}, 2.5) {$.$};
    }
  \draw[line width=2, |-|]
  ({0.5}, 2.5) -- ({4.5}, 2.5) node[anchor=north west] {};
  \node at (({2.5}, 3) {$c_u$};
  \draw[line width=1, |-|]
  ({0.5}, 1.5) -- ({4.5}, 1.5) node[anchor=north west] {};
  \node at ({2.5},1.15) {$c_i$};
            \end{tikzpicture}
            \caption{The vertex-selection day for $u\in V_i$, where $u \notin I$.
            }
        \end{center}
    \end{subfigure}
    \hspace{1cm}
        \begin{subfigure}{0.4\textwidth}
        \begin{center}
            \begin{tikzpicture}
                \draw[line width=2, |-|]
  ({-1}, 1.5) -- ({.48}, 1.5) node[anchor=north west] {};
  \node at (({-.25}, 1.) {$c_0$};
  \node at ({-0.25}, 3.) {};
  \foreach \x in
    {
        1,2,3,4,5,6
    }
    {
        \node at ({-1+\x*0.2}, 2.5) {$.$};
    }
  \draw[line width=1, |-|]
  ({0.5}, 2.5) -- ({4.5}, 2.5) node[anchor=north west] {};
  \node at (({2.5}, 3) {$c_v$};
  \draw[line width=2, |-|]
  ({0.5}, 1.5) -- ({4.5}, 1.5) node[anchor=north west] {};
  \node at ({2.5},1.15) {$c_i$};
            \end{tikzpicture}
            \caption{The vertex-selection day for $v\in V_i$, where $v \in I$.
            }
        \end{center}
    \end{subfigure}

    \vspace{0.5cm}
    
    \begin{subfigure}{\textwidth}
        \begin{center}
            \begin{tikzpicture}[xscale=0.9]

                  \draw[line width=2, |-|]
  ({-1}, 1.5) -- ({.48}, 1.5) node[anchor=north west] {};
  \node at (({-.25}, 1.) {$c_0$};
  \node at ({-0.25}, 3.)  {};
  \foreach \x in
    {
        1,2,3,4,5,6
    }
    {
        \node at ({-1+\x*0.2}, 2.5) {$.$};
    }
    
  \draw[line width=1, |-|]
  ({0.5}, 2.5) -- ({3.48}, 2.5) node[anchor=north west] {};
  \node at (({2.}, 3) {$c_{v_1^i}$};
  \draw[line width=1, |-|]
  ({3.52}, 2.5) -- ({6.48}, 2.5) node[anchor=north west] {};
  \node at (({5.}, 3) {$c_{v_2^i}$};
  \draw[line width=2, |-|]
  ({6.52}, 2.5) -- ({9.48}, 2.5) node[anchor=north west] {};
  \node at (({8}, 3) {$c_{v_3^i}$};
  \draw[line width=1, |-|] ({9.52}, 2.5) -- ({12.5}, 2.5) node[anchor=north west] {};
  \node at (({11.}, 3) {$c_{v_4^i}$};
  
  \draw[line width=2, |-|]
  ({0.5}, 1.5) -- ({1.48}, 1.5) node[anchor=north west] {};
  \node at (1, 1.) {$c_{v_1^i,u_1}$};
  \draw[line width=2, |-|]
  ({1.52}, 1.5) -- ({2.48}, 1.5) node[anchor=north west] {};
  \node at (2, 1.) {$c_{v_1^i,u_2}$};
  \draw[line width=2, |-|]
  ({2.52}, 1.5) -- ({3.48}, 1.5) node[anchor=north west] {};
  \node at (3, 1.) {$c_{v_1^i,u_3}$};
  
  \draw[line width=2, |-|]
  ({3.52}, 1.5) -- ({4.48}, 1.5) node[anchor=north west] {};
  \node at (4, 1.) {$c_{v_2^i,w_1}$};
  \draw[line width=2, |-|]
  ({4.52}, 1.5) -- ({5.48}, 1.5) node[anchor=north west] {};
  \node at (5, 1.) {$c_{v_2^i,w_2}$};
  \draw[line width=2, |-|]
  ({5.52}, 1.5) -- ({6.48}, 1.5) node[anchor=north west] {};
  \node at (6, 1.) {$c_{v_2^i,w_3}$};
  
  \draw[line width=1, |-|]
  ({6.52}, 1.5) -- ({7.48}, 1.5) node[anchor=north west] {};
  \node at (7, 1.) {$c_{v_3^i,y_1}$};
  \draw[line width=1, |-|]
  ({7.52}, 1.5) -- ({8.48}, 1.5) node[anchor=north west] {};
  \node at (8, 1.) {$c_{v_3^i,y_2}$};
  \draw[line width=1, |-|]
  ({8.52}, 1.5) -- ({9.48}, 1.5) node[anchor=north west] {};
  \node at (9, 1.) {{$c_{v_3^i,y_3}$}};
  
  \draw[line width=2, |-|]
  ({9.52}, 1.5) -- ({10.48}, 1.5) node[anchor=north west] {};
  \node at (10, 1.) {{$c_{v_4^i,z_1}$}};
  \draw[line width=2, |-|]
  ({10.52}, 1.5) -- ({11.48}, 1.5) node[anchor=north west] {};
  \node at (11, 1.) {$c_{v_4^i,z_2}$};
  \draw[line width=2, |-|]
  ({11.52}, 1.5) -- ({12.5}, 1.5) node[anchor=north west] {};
  \node at (12, 1.) {$c_{v_4^i,z_3}$};
\end{tikzpicture}
\caption{The $i$-th validation day where $v^i_3\in I$.}
\end{center}

\end{subfigure}
\caption{An example for schedules on the vertex selection days and the validation days given a multicolored independent set $I$. We mark in \textbf{bold} the jobs which are scheduled; the jobs not appearing in the gadget are all identical to the job of client $c_0$. }
\label{fig:vsg}%
\end{figure}
  
The aforementioned \textit{vertex-selection} gadget forces the selection of $\ell$ vertices to satisfy the fairness requirements of the selection clients and vertex clients. Every edge client $c_{v,u}$ can be scheduled on the validation day of $v$'s color as long as $v$ is not selected. Thus, this leaves $r\cdot k$ edge clients that need to be scheduled on the edge days. The fairness requirements of $c^+$ and $c^-$ force the scheduling of one of their jobs on every single edge day. Essentially, the edge days and the validation days form an \textit{incidence-checking} gadget.
Assuming that the jobs of the selected clients form an independent set $I$, there exists no edge $(v,u)$ with $u,v \in I$ whose edge clients must be served on the same day, therefore~$c^+$ and~$c^-$ may be satisfied. See \Cref{fig:edgeday} for illustration.

\begin{figure}[H]
  
    \begin{subfigure}[t]{0.3\textwidth}
        \begin{center}
            \begin{tikzpicture}

    \draw[line width=2, |-|]
  ({0.48}, 1.5) -- ({-1.}, 1.5) node[anchor=north west] {};
  \node at (({-.25}, 1.15) {$c_0$};
  \node at ({-0.25}, 3.) {}; 
  \foreach \x in
    {
        1,2,3,4,5,6
    }
    {
        \node at ({-1+\x*0.2}, 2.5) {$.$};
    }

  \draw[line width=2, |-|]
  ({0.5}, 2.5) -- ({2.48}, 2.5) node[anchor=north west] {};
  \node at (({1.5}, 3) {$c^-$};
  \draw[line width=1, |-|]
  ({1.52}, 1.5) -- ({3.5}, 1.5) node[anchor=north west] {};
  \node at ({2.5},1.15) {$c^+$};
  
  \draw[line width=1, |-|]
  ({.5}, 1.5) -- ({1.48}, 1.5) node[anchor=north west] {};
  \node at ({1.},1.15) {$c_{u,v}$};
  \draw[line width=2, |-|]
  ({2.52}, 2.5) -- ({3.5}, 2.5) node[anchor=north west] {};
  \node at ({3.},3) {$c_{v,u}$};
  
            \end{tikzpicture}
        \end{center}
        \caption{$v\in I , u \notin I$.}
        \end{subfigure}
           \begin{subfigure}[t]{0.3\textwidth}
        \begin{center}
            \begin{tikzpicture}

             \draw[line width=2, |-|]
  ({0.48}, 1.5) -- ({-1.}, 1.5) node[anchor=north west] {};
  \node at (({-.25}, 1.15) {$c_0$};
  \node at ({-0.25}, 3.) {}; 
  \foreach \x in
    {
        1,2,3,4,5,6
    }
    {
        \node at ({-1+\x*0.2}, 2.5) {$.$};
    }
                
  \draw[line width=1, |-|]
  ({0.5}, 2.5) -- ({2.48}, 2.5) node[anchor=north west] {};
  \node at (({1.5}, 3) {$c^-$};
  \draw[line width=2, |-|]
  ({1.52}, 1.5) -- ({3.5}, 1.5) node[anchor=north west] {};
  \node at ({2.5},1.15) {$c^+$};
  
  \draw[line width=2, |-|]
  ({.5}, 1.5) -- ({1.48}, 1.5) node[anchor=north west] {};
  \node at ({1.},1.15) {$c_{u,v}$};
  \draw[line width=1, |-|]
  ({2.52}, 2.5) -- ({3.5}, 2.5) node[anchor=north west] {};
  \node at ({3.},3) {$c_{v,u}$};
  
            \end{tikzpicture}
        \end{center}
        \caption{$ v \notin I, u\in I$.}
        \end{subfigure}       
        \begin{subfigure}[t]{0.3\textwidth}
        \begin{center}
            \begin{tikzpicture}

             \draw[line width=2, |-|]
  ({0.48}, 1.5) -- ({-1.}, 1.5) node[anchor=north west] {};
  \node at (({-.25}, 1.15) {$c_0$};
 \node at ({-0.25}, 3.) {}; 
  \foreach \x in
    {
        1,2,3,4,5,6
    }
    {
        \node at ({-1+\x*0.2}, 2.5) {$.$};
    }
                
  \draw[line width=2, densely dash dot, |-|]
  ({0.5}, 2.5) -- ({2.48}, 2.5) node[anchor=north west] {};
  \node at (({1.5}, 3) {$c^-$};
  \draw[line width=2, densely dash dot, |-|]
  ({1.52}, 1.5) -- ({3.5}, 1.5) node[anchor=north west] {};
  \node at ({2.5},1.15) {$c^+$};
  
  \draw[line width=1, |-|]
  ({.5}, 1.5) -- ({1.48}, 1.5) node[anchor=north west] {};
  \node at ({1.},1.15) {$c_{u,v}$};
  \draw[line width=1, |-|]
  ({2.52}, 2.5) -- ({3.5}, 2.5) node[anchor=north west] {};
  \node at ({3.},3) {$c_{v,u}$};
  
\end{tikzpicture}
\end{center}
\caption{$u,v \notin I$: schedule either $c^+$ or\ $c^-$ such that $c^+$ and $c^-$ are satisfied. }
\end{subfigure}

\caption{
The three possible cases on an \textit{edge day} (with $v\in V_i$ and $u\in V_j$ with $i<j$) given a multicolored independent set $I$. We mark the jobs which are scheduled in a schedule constructed from $I$ in \textbf{bold}, and the jobs which do not appear in the figure all have the same job as that of client $c_0$ (but are not scheduled). There are $k\cdot r$ days of types (a) and~(b), as $|E|$ is even can freely use (c) days to ensure both $c^+$ and $c^-$ are satisfied.  
}
\label{fig:edgeday}
\end{figure}
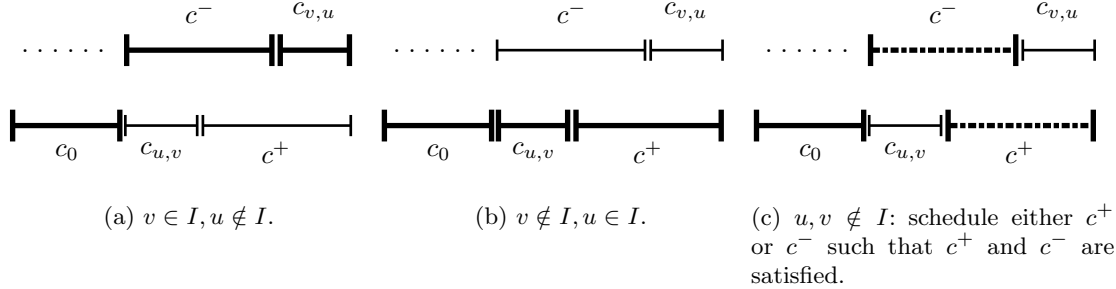

After having described the reduction, we continue by showing that the overall conflict graph of the constructed instance indeed has bounded treewidth and that it is correct.
We can construct a tree-decomposition~$(\mathcal{X},F)$ of width four as follows: 

\begin{figure}[H]
    \centering
    
\begin{minipage}{0.5\textwidth}
Nodes $\mathcal{X}$:
\begin{itemize}
    \item $X_{c_0} = \{ c_0,c^+,c^-\}$
    \item $\forall i\in \{1,\dots, \ell\}\ :\ X_{S_i} = \{c_0 ,c^+,c^-,c_i \}$
    \item $\forall v\in V\ :\ X_v = \{c_0 ,c^+,c^-,c_v, c_i \}$
    \item $\forall u\in N(v)\ :\ X_{v,u} = 
    \{ c_0,c^+,c^-,c_v, c_{v,u}\}$
\end{itemize}

Edges $F = F_1\cup F_2 \cup F_3$:
\begin{itemize}
    \item $F_1 = \{ (X_{c_0},X_{S_i})\ |\ i\in \{1,\dots,\ell\} \}$
    \item $F_2 = \{ (X_{S_i},X_{v})\ |\ v \in V_i \}$
    \item $F_3 = \{(X_v, X_{v,u}) \ |\ v\in V \wedge u \in N(v) \}$
\end{itemize}
\end{minipage}%
\hfill%
\begin{minipage}{0.5\textwidth}

\scalebox{.75}{

\tikzset{level 1 concept/.append style={ sibling angle=45,level distance = 45mm}}
\tikzset{level 2 concept/.append style={ sibling angle=45,level distance = 25mm}}
\tikzset{level 3 concept/.append style={ sibling angle=45,level distance = 20mm}}
\tikzset{level 3 unimportant/.append style={ sibling angle=45,level distance = 20mm}}

\begin{tikzpicture}

\newcommand{\colorxdb}{black}
\newcommand{\colorxsb}{black}
\newcommand{\colorxvb}{black}
\newcommand{\colorxvub}{black}

  \path[mindmap,concept color=\colorxdb,text=black,]
    node[draw=gray,circle,fill=white!50,minimum size=.1cm,scale=0.78] {\Huge $\boldsymbol{X_{c_0}}$}
    [clockwise from=90]
    child[level distance=40mm, concept color=\colorxsb] {
      node[draw=gray,circle,fill=white!50,minimum size=.1cm,scale=0.78] {\Huge $\boldsymbol{X_{S_{1}}}$}
      [clockwise from=115]
      child[level distance=25mm,sibling angle=25,concept color=\colorxvb] { node[draw=gray,circle,fill=white!50,minimum size=.1cm, scale=0.3] {}
          [clockwise from=140]
          child[level distance=12.5mm,sibling angle=20,concept color=\colorxvub] { node[draw=gray,circle,fill=white!50,minimum size=.1cm, scale=0.3] {} }
          child[level distance=12.5mm,sibling angle=20,concept color=\colorxvub] { node[draw=gray,circle,fill=white!50,minimum size=.1cm, scale=0.3] {} }
          child[level distance=12.5mm,sibling angle=20,concept color=\colorxvub] { node[draw=gray,circle,fill=white!50,minimum size=.1cm, scale=0.3] {} }
      }
      child[level distance=25mm,sibling angle=25, concept color=\colorxvb] { node[draw=gray,circle,fill=white!50,minimum size=.1cm, scale=0.3] {}
          [clockwise from=110.]
          child[level distance=12.5mm,sibling angle=20,concept color=\colorxvub] { node[draw=gray,circle,fill=white!50,minimum size=.1cm, scale=0.3] {} }
          child[level distance=12.5mm,sibling angle=20,concept color=\colorxvub] { node[draw=gray,circle,fill=white!50,minimum size=.1cm, scale=0.3] {} }
          child[level distance=12.5mm,sibling angle=20,concept color=\colorxvub] { node[draw=gray,circle,fill=white!50,minimum size=.1cm, scale=0.3] {} }
      }
      child[level distance=25mm,sibling angle=25, concept color=\colorxvb] { node[draw=gray,circle,fill=white!50,minimum size=.1cm, scale=0.3] {}
          [clockwise from=80]
          child[level distance=12.5mm,sibling angle=20,concept color=\colorxvub] { node[draw=gray,circle,fill=white!50,minimum size=.1cm, scale=0.3] {} }
          child[level distance=12.5mm,sibling angle=20,concept color=\colorxvub] { node[draw=gray,circle,fill=white!50,minimum size=.1cm, scale=0.3] {} }
          child[level distance=12.5mm,sibling angle=20,concept color=\colorxvub] { node[draw=gray,circle,fill=white!50,minimum size=.1cm, scale=0.3] {} }
      }
    }
    child[concept color=\colorxsb] {
      node[draw=gray,circle,fill=white!50,minimum size=.1cm,scale=0.78] {\Huge $\boldsymbol{X_{S_{i}}}$}
      [clockwise from=90]
      child[,concept color=\colorxvb] { node[draw=gray,circle,fill=white!50,minimum size=.1cm, scale=0.78] {\Huge$\boldsymbol{X_{v_{1}}}$} 
          [clockwise from=107.5]
          child[level distance=17.5mm,sibling angle=17.5,concept color=\colorxvub] { node[draw=gray,circle,fill=white!50,minimum size=.1cm, scale=0.3] {} }
          child[level distance=17.5mm,sibling angle=17.5,concept color=\colorxvub] { node[draw=gray,circle,fill=white!50,minimum size=.1cm, scale=0.3] {} }
          child[level distance=17.5mm,sibling angle=17.5,concept color=\colorxvub] { node[draw=gray,circle,fill=white!50,minimum size=.1cm, scale=0.3] {} }
      }
      child[concept color=\colorxvb] {
        node[draw=gray,circle,fill=white!50,minimum size=.1cm,scale=0.78, clockwise from=90] {\Huge $\boldsymbol{X_{v_p}}$}
          [clockwise from=90]
          child[concept color=\colorxvub] { node[draw=gray,circle,fill=white!50,minimum size=.1cm, scale=0.78] {} }
          child[level distance=25mm, concept color=\colorxvub] { node[draw=gray,circle,fill=white!50,minimum size=.1cm, scale=.78, minimum size = 1.5cm,text width=2.1cm, align=center] {
        \huge$\boldsymbol{X_{v_{p},u}}$
          } }
          child[concept color=\colorxvub] { node[draw=gray,circle,fill=white!50,minimum size=.1cm, scale=1.03] {} }
    }
      child[,concept color=\colorxvb] {
          node[ draw=gray,circle,fill=white!50,minimum size=2.05cm, scale=0.85] {\Huge $\boldsymbol{X_{v_{n}}} $}
          [clockwise from=17.5]
          child[level distance=17.5mm,sibling angle=17.5,concept color=\colorxvub] { node[draw=gray,circle,fill=white!50,minimum size=.1cm, scale=0.3] {} }
          child[level distance=17.5mm,sibling angle=17.5,concept color=\colorxvub] { node[draw=gray,circle,fill=white!50,minimum size=.1cm, scale=0.3] {} }
          child[level distance=17.5mm,sibling angle=17.5,concept color=\colorxvub] { node[draw=gray,circle,fill=white!50,minimum size=.1cm, scale=0.3] {} }
    }
    }
    child[level distance=40mm, concept color=\colorxsb] {
      node[draw=gray,circle,fill=white!50,minimum size=.1cm, scale=0.78] {\Huge $\boldsymbol{X_{S_{\ell}}}$}
      [clockwise from=25]
      child[level distance=25mm,sibling angle=25,concept color=\colorxvb] { node[draw=gray,circle,fill=white!50,minimum size=.1cm, scale=0.3] {}
          [clockwise from=45]
          child[level distance=12.5mm,sibling angle=20,concept color=\colorxvub] { node[draw=gray,circle,fill=white!50,minimum size=.1cm, scale=0.3] {} }
          child[level distance=12.5mm,sibling angle=20,concept color=\colorxvub] { node[draw=gray,circle,fill=white!50,minimum size=.1cm, scale=0.3] {} }
          child[level distance=12.5mm,sibling angle=20,concept color=\colorxvub] { node[draw=gray,circle,fill=white!50,minimum size=.1cm, scale=0.3] {} }
      }
      child[level distance=25mm,sibling angle=25,concept color=\colorxvb] { node[draw=gray,circle,fill=white!50,minimum size=.1cm, scale=0.3] {}
          [clockwise from=20.]
          child[level distance=12.5mm,sibling angle=20,concept color=\colorxvub] { node[draw=gray,circle,fill=white!50,minimum size=.1cm, scale=0.3] {} }
          child[level distance=12.5mm,sibling angle=20,concept color=\colorxvub] { node[draw=gray,circle,fill=white!50,minimum size=.1cm, scale=0.3] {} }
          child[level distance=12.5mm,sibling angle=20,concept color=\colorxvub] { node[draw=gray,circle,fill=white!50,minimum size=.1cm, scale=0.3] {} }
      }
      child[level distance=25mm,sibling angle=25,concept color=\colorxvb] { node[draw=gray,circle,fill=white!50,minimum size=.1cm, scale=.3] {}
          [clockwise from=-10]
          child[level distance=12.5mm,sibling angle=20,concept color=\colorxvub] { node[draw=gray,circle,fill=white!50,minimum size=.1cm, scale=0.3] {} }
          child[level distance=12.5mm,sibling angle=20,concept color=\colorxvub] { node[draw=gray,circle,fill=white!50,minimum size=.1cm, scale=0.3] {} }
          child[level distance=12.5mm,sibling angle=20,concept color=\colorxvub] { node[draw=gray,circle,fill=white!50,minimum size=.1cm, scale=0.3] {} }
      }
    };

\node[rotate=-80] at (3.4,1.5) {\Huge$\dots$};
\node[rotate=-10] at (1.65,3.5) {\Huge$\dots$};
\node[rotate=-10] at (3.8,4.7) {\huge$\dots$};
\node[rotate=-80] at (4.7,3.75) {\huge$\dots$};
\node[rotate=-10] at (5.5,6.2) {\large$\dots$};
\node[rotate=-80] at (6.2,5.5) {\large$\dots$};

\end{tikzpicture}
}

\end{minipage}
\caption{A tree decomposition $(\mathcal{X},F)$ of the overall conflict graph constructed from the instance of \WEQUIJIT. Observe that the width of this decomposition is $|X_v|-1 = 4$.}
\label{fig:tw}
\end{figure}
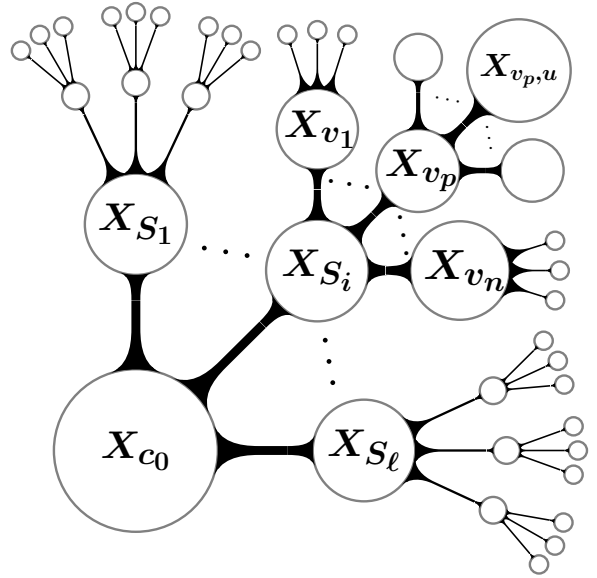

$(\Rightarrow)$
We next show that if $G$ contains a multicolored independent set $I \subseteq V_1 \cup \dots  \cup V_\ell$ of size $\ell$ then our constructed \WEQUIJIT\ instance has a solution in which each client's fairness parameter is met. We construct such a solution as follows: On each day, the dummy client~$c_0$ is scheduled. For each~${i \in \{1,\dots,\ell\}}$, the clients from the vertex-selection gadget are scheduled as follows: For each vertex~$v = v^i_p \in V_i \setminus I$, we schedule the job of $c_v$ on the $p$'th vertex day of the $i$'th color. For $v = v^i_p \in V_i \cap I$, we schedule the selection client~$c_i$ on the $p$'th vertex day of the $i$'th color and the job of $c_v$ on $i$'th validation day. Additionally, on the $i$'th validation day, all clients~$c_{v,u}$ are scheduled for each $v \in V_i \setminus I$. Finally, for each edge $e_p=(v,u)$ or $e = (u, v)$ incident to $v \in V_i \cap I$, we schedule client~$c_{v,u}$ on the $p$'th edge day. Note that by the construction, all these jobs can be scheduled on their appropriate days without conflicts. It remains to schedule the interaction clients~$c^+$ and~$c^-$. For each edge day of edge~$(u, v)$ with $u \in I$ or $v \in I$, we schedule the job of~$c^+$ if $u \in I$, otherwise we schedule the job of~$c^-$. Note that this does not lead to a conflict as $I$ is an independent set. On the remaining edge days, no client has been scheduled so far (except for the dummy client), and we schedule either $  c^+$ or $c^-$ in a way such that in the end, $c^+$ and $c^-$ are scheduled equally often. All together, the interaction clients~$c^+$ and $c^-$ are scheduled on $|E|/2$ days each, the dummy client is scheduled on every day, and all other clients are scheduled on a single day.

$(\Leftarrow)$
Assume that there is a schedule to the constructed \WEQUIJIT\ instance in which the fairness parameter of each client is met. Then the jobs of the dummy client are scheduled on each day, and so each selection client $c_i$ has to have his job scheduled on some vertex day. Let $\{p_1, \ldots,  p_\ell\}$ be a set of indices such that for each $i \in \{1,\dots,\ell\}$, the job of selection client~$c_i$ is scheduled at the $p_i$'th vertex day of the color~$i$. Consider some $v=v^{p_i}_i\in V_i$. By construction, client~$c_v$ needs to be scheduled on the $i$'th validation day as otherwise the job of the dummy client cannot be scheduled. This implies that a job of the edge client~$c_{v,u}$ is scheduled on~$(v,u)$'s edge day for each neighbor~$u$ of $v$ in $G$. The jobs of the interaction clients~$c^+$ and $c^-$ can only be scheduled on the edge days, and at most one interaction client job per day (as $c^+$ and $c^-$ conflict on each day). As $k^+ + k^- = |E|$, it follows that on each edge day, exactly one of~$c^+$ and $c^-$ is scheduled. Consequently, there is no edge~$(u, v)$ such that on $(u, v)$'s edge day the jobs of both~$c_{u,v}$ and $c_{v,u}$ are scheduled (as $c^+$ conflicts with $c_{v,u}$ and $c^-$ conflicts with $c_{u,v}$). It follows that $u$ and $v$ are non-adjacent in $G$, and so $\{v_{p_1}, \ldots, v_{p_{\ell}}\}$ is a multicolored independent set in $G$.
\end{proof}

\begin{lemma}
\label{lem:td}%
There is a polynomial time reduction from \WEQUIJIT\ to \EQUIJIT\ which increases the treewidth by at most 2.
\end{lemma}

\begin{proof}
Let $\mathcal{I}$ be an instance of \WEQUIJIT\ with $\q$ days and $n$ clients. We construct an instance $\mathcal{I}'$ (illustrated in \Cref{fig:reduction-unweighted}) of \EQUIJIT\ as follows: Starting with $\mathcal{I}$, we append another $\q$ days, add two clients~$c^-$ and $c^+$, and set $k := \q$. The jobs of $c^-$ and $c^+$ have the same processing time and due date, so they cannot both be scheduled on the same day. On the first ``original'' $\q$ days, none of the clients from~$\mathcal{I}$ intersects with clients~$c^-$ and $c^+$. On the last ``additional'' $\q$ days, no two different clients from~$\mathcal{I}$ conflict. For each client~$c$ from $\mathcal{I}$, there are exactly $k_c$ days where client~$c$ intersects with $c^-$ and $c^+$. More formally, let $d_{\max}$ be the maximum due date in $\mathcal{I}$. On each day of the $2\q$ days in~$\mathcal{I}'$, clients~$c^-$ and $c^+$ have processing time $n$ and due date $d_{\max} + n$. On days~$1$ to~$\q$, each client from~$\mathcal{I}$ has the same processing time and due date as in~$\mathcal{I}$.
Let $c_1, \ldots, c_n$ be the clients from~$\mathcal{I}$ such that $k_{c_i}\leq k_{c_j}$ for $i<j$ (ordering the clients in this way is not needed for correctness, but makes the illustration in \Cref{fig:reduction-unweighted} nicer). Client~$c_i$ has processing time~$1$ and due date $d_{\max} + i$ on days $\q+1, \ldots, \q+k_c$ and processing time~$1$ and due date $d_{\max} + n + i$ on days $\q+k_c + 1 , \ldots, 2\q$.

\begin{figure}[H]
\centering
\begin{subfigure}{0.8\textwidth}
\begin{center}
\begin{tikzpicture}

    \def\dmax{5.5}
    \def\eps{.1}
    \draw[thick,->] (0,0) -- (14,0) node[anchor=north west] {};
    \draw[line width=0.1, dashed]
    (\dmax,-.) -- (\dmax,2.2) node[anchor=north] {};
    \node at (({\dmax}, -.5) {$d_{\max}$};
    \draw[line width=0.1, dashed]
    (0.5,-.) -- (0.5,2.2) node[anchor=north] {};
    \node at (({11.}, -.5) {$d_{\max}+n$};
    
    \draw[line width=0.1, dashed]
    (11.,-.) -- (11.,2.2) node[anchor=north] {};

    \draw[line width=.75, >-<]
  ({0.5+\eps}, 2.) -- ({\dmax-\eps}, 2.) node[anchor=north west] {};
  \node at (({\dmax*0.5+\eps*0.5}, 2.2) {The $i$-th day of $\mathcal{I}$};
    

    \node at (({\dmax*0.5+0.25}, 1.1) {$\dots$};
  \draw[line width=1, |-|]
  ({0.5+\eps}, .5) -- ({2.3}, .5) node[anchor=north west] {};
  \draw[line width=1, |-|]
  ({0.5+\eps}, 1.) -- ({1.5-\eps}, 1.) node[anchor=north west] {};
  \draw[line width=1, |-|]
  ({1.5+\eps}, 1.) -- ({2.5-\eps}, 1.) node[anchor=north west] {};
  \draw[line width=1, |-|]
  (1, .75) -- ({2.}, .75) node[anchor=north west] {};
  \draw[line width=1, |-|]
  (1, 1.3) -- ({2.5}, 1.3) node[anchor=north west] {};
  \draw[line width=1, |-|]
  (1, 1.3) -- ({2.5}, 1.3) node[anchor=north west] {};

\draw[line width=1, |-|]
  (\dmax*0.5+1.7, .3) -- ({\dmax*0.5+2.5}, .3) node[anchor=north west] {};
  \draw[line width=1, |-|]
  ({\dmax*0.5+0.75+\eps}, .5) -- ({\dmax*0.5+1.5}, .5) node[anchor=north west] {};
  \draw[line width=1, |-|]
  ({\dmax*0.5+1.75+\eps}, .75) -- ({\dmax*0.5+2.65-\eps}, .75) node[anchor=north west] {};
  \draw[line width=1, |-|]
  (\dmax*0.5+.75, .75) -- ({\dmax*0.5+1.7}, .75) node[anchor=north west] {};
    \draw[line width=1, |-|]
  ({\dmax*0.5+1.+\eps}, 1.) -- ({\dmax*0.5+1.5-\eps}, 1.) node[anchor=north west] {};
    \draw[line width=1, |-|]
  ({\dmax*0.5+1.65+\eps}, 1.) -- ({\dmax*0.5+2.4-\eps}, 1.) node[anchor=north west] {};
  \draw[line width=1, |-|]
  (\dmax*0.5+1, 1.3) -- ({\dmax*0.5+2.}, 1.3) node[anchor=north west] {};

   \draw[line width=1, |-|]
  (\dmax+\eps, 1.3) -- ({11.-\eps}, 1.3) node[anchor=north west] {};
  \node at (({8.5}, 1.6) {$c^+$};
   \draw[line width=1, |-|]
  (\dmax+\eps, .5) -- ({11.-\eps}, .5) node[anchor=north west] {};
  \node at (({8.5}, .7) {$c^-$};
  
    \end{tikzpicture}
    \caption{When $1\leq i\leq \q$ the interval, the jobs of the original clients have the same intervals as in  the $i$-th day of $\mathcal{I}$.}
    \end{center}
    \end{subfigure}

 \begin{subfigure}{0.8\textwidth}
        \begin{center}
            \begin{tikzpicture}

                 \def\dmax{5.5}
    \def\eps{.1}
    \draw[thick,->] (0,0) -- (14,0) node[anchor=north west] {};
    \draw[line width=0.1, dashed]
    (\dmax,-.) -- (\dmax,2.2) node[anchor=north] {};
    \node at (({\dmax}, -.5) {$d_{\max}$};
    \node at (({11.}, -.5) {$d_{\max}+n$};
    \draw[line width=0.1, dashed]
    (11.,-.) -- (11.,2.2) node[anchor=north] {};

   \draw[line width=1, |-|]
  (\dmax+\eps, 1.3) -- ({11.-\eps}, 1.3) node[anchor=north west] {};
  \node at (({8.5}, 1.6) {$c^+$};
   \draw[line width=1, |-|]
  (\dmax+\eps, .5) -- ({11.-\eps}, .5) node[anchor=north west] {};
  \node at (({8.5}, .7) {$c^-$};

\foreach \x/\i in
    {
        1/1,2/3,3/1
    }
    {
        \ifthenelse{\i = 0}{
        
        }{
         \ifthenelse{\i = 1}{
                \draw[line width=1, |-|]
                  (\dmax+3+\x*0.7, 2) -- ({\dmax+3+\x*0.7+0.3}, 2) node[anchor=north west] {};
         }{

            \node at (\dmax+3+\x*0.7+0.15, 2) {$\dots$};
         
         }
        }

    }  

\foreach \x/\i in
    {
        1/1,2/3,3/1
    }
    {
        \ifthenelse{\i = 0}{
        
        }{
         \ifthenelse{\i = 1}{
                \draw[line width=1, |-|]
                  (10.5+\x*0.7, 2) -- ({10.5+\x*0.7+0.3}, 2) node[anchor=north west] {};
         }{

            \node at (10.5+\x*0.7+0.15, 2) {$\dots$};
         
         }
        }

    }

\end{tikzpicture}
\caption{When $\q< i\leq2\q$ let the jobs of original clients either intersect the jobs of $c^-$ and $c^+$  (for $\{\ c\ |\ k_c \leq i\ \}$) \textit{or} intersect no jobs at all (for $\{\ c\ |\ k_c > i\ \}$).}
\end{center}
\end{subfigure}
\caption{An intersection model of the $i$-th day's jobs.}
\label{fig:reduction-unweighted}
\end{figure}
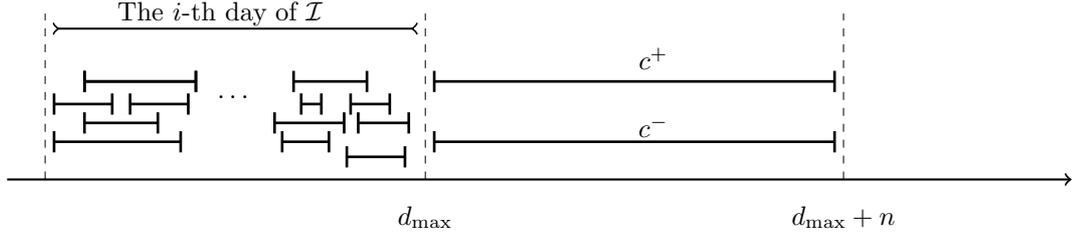
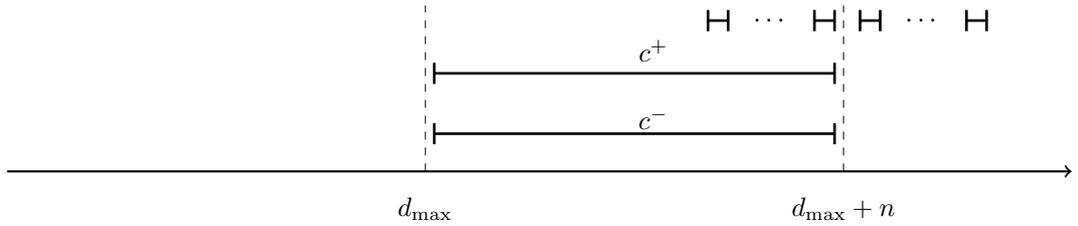

As we only add two clients (and no new conflicts between clients from~$\mathcal{I}$), the treewidth of the overall conflict graph increases by at most two. What remains is to show the correctness of our construction. Given a solution to~$\mathcal{I}'$, we construct a solution to~$\mathcal{I}$ by deleting clients~$c^-$ and $c^+$ and days~$\q+1, \ldots, 2\q$. Note that on each day, either $c^-$ or $c^+$ needs to be scheduled as they conflict on each day. Thus, each client~$c$ from $\mathcal{I}$ cannot be scheduled on days~$\q+1, \ldots, m+k_c$, implying that~$c$ is scheduled at least $k_c$ times on days~$1, \ldots, \q$. Thus, we have a solution to~$\mathcal{I}$. Conversely, given a solution to~$\mathcal{I}$, we construct a solution to~$\mathcal{I}'$ as follows: Client~$c^-$ is scheduled on the first $\q$ days, while $c^+$ is scheduled on the last $\q$ days (days~$\q+1, \ldots, 2\q$). On days~$1, \ldots, \q$, the clients from~$\mathcal{I}$ are scheduled as in the solution to~$\mathcal{I}$. Furthermore, every client~$c$ from~$\mathcal{I}$ is scheduled on days~$\q+k_c+1, \ldots, 2\q$. Thus, client~$c$ is scheduled on $k_c + \q-k_c = \q$ days, so we get a solution to~$\mathcal{I}'$.
\end{proof}

\begin{theorem}
\label{thm:treedepth}
The \EQUIJIT\ problem is NP-hard for $\tau \geq 6$.
\end{theorem}
\begin{proof}
In \Cref{lem:td-weighted} we described a reduction from the NP-hard \textsc{Multicolored Independent Set} to an instance of \WEQUIJIT\ with a treewidth 4. \Cref{lem:td} describes a reduction from \WEQUIJIT\ to \EQUIJIT\ by which the treewidth of the overall conflict graph increases by at most 2. The theorem thus follows.
\end{proof}

\subsection{FPT algorithm with respect to the number of days and the treewidth}\label{sec:fptqw}

We next show an FPT algorithm for parameter $\q+\tau$, number of days $\q$ plus the treewidth of the overall conflict graph. We will use a simplified version of tree decompositions, called \emph{nice tree decompositions}, introduced by Kloks~\cite{kloks1994treewidth}.  
\begin{definition}
A \textit{nice tree decomposition} is a tree decomposition $\mathcal{T}=(\mathcal{X},F)$ where the following conditions additionally hold:
\begin{enumerate}
\item Every node $X \in \mathcal{X}$ has at most two children. 
\item If node $X$ has two children $X'$ and $X''$ then $X=X'=X''$, and $X$ is called a \textit{join node}. 
\item If node $X$ has a single child $X'$ then either of the two following holds:
\begin{enumerate}
\item $X = X' \setminus \{j\}$ for some client $j \in X'$. In this case, we call $X$ a \textit{forget node}.
\item $X = X' \cup \{j\}$ for some client $j \in \{1,\ldots,n\} \setminus X'$. In this case, we say that $X$ is an \textit{introduce node}.
\end{enumerate}
\end{enumerate}
\end{definition}
\noindent It is known that a nice tree decomposition of width $\tau=\tau(G)$ of a graph $G$ can be computed from a given tree decomposition of equal width in linear time~\cite{alber2002improved,kloks1994treewidth}.

Let $\mathcal{I}$ be an instance of \EQUIJIT\, and let $\mathcal{T}$ be a given nice tree decomposition of the overall conflict graph $G=(\{1,\ldots,n\},E_1 \cup \cdots \cup E_m\})$ of $\mathcal{I}$. To present our algorithm, we will need additional terminology. Let $V=\{1,\ldots,n\}$ denote the set of $n$ clients of $\mathcal{I}$. Recall that a (not necessarily feasible) schedule $\sigma$ for $V$ is a tuple $\sigma = (\sigma_1,\ldots,\sigma_m)$ of $m$ subsets of $V$; that is,~$\sigma_i \subseteq V$ for each day~$i \in \{1,\ldots,m\}$. 
For a subset of clients $X \subseteq V$, we let $\sigma \cap X$ denote the schedule $(\sigma'_1 \cap X,\ldots,\sigma'_m \cap X)$.

We will be particularly interested in schedules that are defined only on subsets of clients. Let $X \subseteq V$. We say that a schedule \emph{$\sigma=(\sigma_1,\ldots,\sigma_m)$ is defined on~$X$} if $\sigma_i \subseteq X$ for each day~$i \in \{1,\ldots,m\}$. We say that $\sigma$ is \emph{feasible and fair on~$X$} if $\sigma$ is feasible, defined on~$X$, and we have $\sum_j Z_{i,j} \geq k$ for each client~$j \in X$. That is, a feasible and fair schedule on $X$ is a feasible schedule which is defined only on the clients of $X$, and in which each client of $X$ has at least $k$ of its jobs scheduled among all $m$ days. We let $\Sigma(X)$ denote the set of all feasible and fair schedules on $X$. The following easy lemma will be crucial for our algorithm:
\begin{lemma}
\label{lem:TWDP}
Let $X \subseteq V$ be a subset of at most $\tau + 1$ clients. Then $|\Sigma(X)| = O(2^{\tau m})$, and we can compute~$\Sigma(X)$ in $2^{O(\tau m)}$ time.
\end{lemma}

Let $X \in \mathcal{X}$ be a node of the given nice tree decomposition $\mathcal{T}$. We let~$\mathcal{T}_X$ denote the subtree of~$\mathcal{T}$ rooted at $X$, \emph{i.e.}, the subtree of~$\mathcal{T}$ that contains $X$ and all of its decedents in $\mathcal{T}$. We use~ $V(\mathcal{T}_X)$ to denote the set of all clients that are in some node of $X'$ of $\mathcal{T}(X)$. In other words, $V(\mathcal{T}_X)$ is the set of all clients in $X$, and all clients in any descendent of $X$ in $\mathcal{T}$. We compute a dynamic programming table~$T$ on the nodes of~$\mathcal{T}$, in bottom-up fashion, to determine whether the set $\Sigma(V)$ is non-empty. For any node $X$ of $\mathcal{T}$, and any schedule $\sigma \in \Sigma(X)$, we have a corresponding entry~$T[X,\sigma]$ in~$T$ where the following invariant will hold: 
$$
T[X,\sigma] = 1 \iff \sigma = \sigma^* \cap X \text{ for some } \sigma^* \in \Sigma(V(\mathcal{T}_X)) 
$$
Equivalently, $T[X,\sigma] = 1$ if and only if $\sigma$ can be extended to some feasible and fair schedule $\sigma^*$ for $V(\mathcal{T}_X)$. In this way, there exists a solution to the entire \EQUIJIT\ instance if and only if $T[X,\sigma] = 1$ for the root node $X$ of $\mathcal{T}$ and some $\sigma$, as $V(\mathcal{T}_X) = V$ in this case. We proceed to show how to compute all four types of nodes of $\mathcal{T}$, starting with leaf nodes which are the base case of the dynamic program.   

\paragraph*{Leaf nodes.} 

For leaf node $X$ of $\mathcal{T}$ we have $V(\mathcal{T}_X) = X$ by definition. Thus, we have $T[X,\sigma]=1$ for all $\sigma \in \Sigma(X)$. We can therefore compute all entries in~$T$ corresponding to~$X$ in $2^{O(\tau m)}$ time using \Cref{lem:TWDP}. 

\paragraph*{Forget nodes.} 

Let $X$ be a forget node in $\mathcal{T}$, and let $X'\in \mathcal{X}$ be the child of $X$ in~$\mathcal{T}$. Then we have $X=X' \setminus \{j\}$ for some client $j \in \{1, \ldots,n\}$. Given a schedule $\sigma \in \Sigma(X)$, we compute the entry~$T[X,\sigma]$ by
$$
T[X,\sigma] = 
\begin{cases}
1  & T[X',\sigma'] = 1  \text{ for some } \sigma' \text{ with }\sigma=\sigma' \cap X;\\
0  & \text{otherwise.}
\end{cases} 
$$
Note that this computation can be done in $2^{O(\tau m)}$ time, and it correctness follows from the fact that by definition we have $V(\mathcal{T}_X) = V(\mathcal{T}_{X'})$. Thus, $T[X',\sigma']=1$ if and only if $T[X,\sigma' \cap X]=1$ for all $\sigma' \in \Sigma(X')$.  

\paragraph*{Introduce nodes.} 

Let $X$ be an introduce node in $\mathcal{T}$, and let $X'\in \mathcal{X}$ be the child of $X$ in~$\mathcal{T}$. Then $X=X' \cup \{j\}$ for some client $j \in \{1, \ldots,n\}$. For a schedule $\sigma \in \Sigma(X)$, we compute the entry~$T[X,\sigma]$ by
$$
T[X,\sigma] 
= \begin{cases}
1  & T[X',\sigma']=1 \text{ where } \sigma'=\sigma \cap X';\\
0  & \text{otherwise.}
\end{cases} 
$$
Observe that as $X= X' \cup \{j\}$ we have  $V(\mathcal{T}_X) = V(\mathcal{T}_{X'}) \cup \{j\}$. Moreover, by requirements of a tree decomposition, any client $j_0 \in V(\mathcal{T}_{X'}) \setminus X'$ is not adjacent to client $j$ in the overall conflict graph. This means that any such client~$j_0$ does not have any conflicting jobs with the jobs of client~$j$ on any of the $m$ days. Now consider any schedule $\sigma' = (\sigma'_1,\ldots,\sigma'_m) \in \Sigma(X')$ with $T[X',\sigma']=1$, and let $\sigma^* = (\sigma^*_1,\ldots,\sigma^*_m) \in \Sigma(V(\mathcal{T}_{X'}))$ be such that $\sigma'=\sigma^* \cap X'$. It follows that $\sigma^{**}=(\sigma^*_1 \cup \sigma_1,\ldots,\sigma^*_m \cup \sigma_m)$ is feasible, and so $\sigma^{**} \in \Sigma(V(\mathcal{T}_X))$. Thus, as $\sigma = \sigma^{**} \cap X$, we have $T[X,\sigma]=1$. Conversely, for any $\sigma \in \Sigma(X)$ for which there exists some $\sigma^{**} \in \Sigma(V(\mathcal{T}_X))$ with $\sigma = \sigma^{**} \cap X$, we have by definition that $\sigma^* = \sigma^{**} \cap X'$ is in $\Sigma(V(\mathcal{T}_{X'}))$, and that $\sigma^* \cap X' = \sigma \cap X' \in \Sigma(X')$. Therefore, we have~$T[X',\sigma \cap X']=1$.  

\paragraph*{Join nodes.} 

Let $X$ be a join node in $\mathcal{T}$, and let $X',X''\in \mathcal{X}$ be the children of $X$ in~$\mathcal{T}$. Then $X=X'=X''$. For a schedule $\sigma \in \Sigma(X)$, we compute the entry~$T[X,\sigma]$ by
$$
T[X,\sigma] 
= \begin{cases}
1  & T[X',\sigma]=1 \text{ and } T[X'',\sigma]=1; \\
0  & \text{otherwise.}
\end{cases} 
$$
Note that we have $V(\mathcal{T}_X) = V(\mathcal{T}_{X'}) \cup V(\mathcal{T}_{X''})$. Thus, for any schedule $\sigma^* \in \Sigma(V(\mathcal{T}_X))$ we have $\sigma^* \cap V(\mathcal{T}_{X'}) \in \Sigma(V(\mathcal{T}_{X'}))$ and $\sigma^* \cap V(\mathcal{T}_{X''}) \in \Sigma(V(\mathcal{T}_{X''}))$. Moreover, by requirements of a tree decomposition, any client $j_1 \in V(\mathcal{T}_{X'}) \setminus X$ is not adjacent to any client $j_2 \in V(\mathcal{T}_{X''}) \setminus X$. Thus, for any pair of schedules $(\sigma^*_1,\ldots,\sigma^*_m) \in \Sigma(V(\mathcal{T}_{X'}))$ and $(\sigma^{**}_1,\ldots,\sigma^{**}_m) \in \Sigma(V(\mathcal{T}_{X''}))$, we have that $(\sigma^*_1 \cup \sigma^{**}_1,\ldots,\sigma^*_m \cup \sigma^{**}_m) \in \Sigma(V(\mathcal{T}_{X}))$. It follows that $T[X,\sigma]=1$ if and only if $T[X',\sigma]=T[X'',\sigma]=1$ for all schedules~$\sigma \in \Sigma(X)$.

In conclusion, it is easy to see that from the above that we can compute the value each entry $T[X,\sigma]$ in our dynamic programming table in $2^{O(\tau m)}$ time. As there are $2^{O(\tau m)} \cdot n$ such entries, it follows that the total running time of our algorithm can be bounded by $2^{O(\tau m)} \cdot n$. 

\begin{theorem}
\label{thm:fpt:w+m}%
\EQUIJIT\ is solvable in $2^{O(\tau m)} \cdot n$ time.
\end{theorem}

\subsection{FPT algorithm with respect to the number of clients}\label{sec:fptn}

We next present an FPT algorithm for \EQUIJIT\ with respect to the number of clients parameter~$n$. Our algorithm exploits the fact that when $\q$ is large enough ($\q \in 2^{\omega(n \log n)}$), some daily conflict graphs are guaranteed to occur multiple times. Specifically, we provide an ILP formulation where the number of variables is bounded by a function of $n$, which directly implies an FPT algorithm with respect to $n$ due to Lenstra's Theorem~\cite{lenstra1983integer}.

\begin{theorem}[Lenstra's ILP Algorithm~\cite{lenstra1983integer}]
\label{thm:Lenstra}%
There exists an algorithm for solving any ILP on $n$ variables and $m$ constraints in $n^{O(n)}m^{O(1)}$ time.
\end{theorem}

\begin{theorem}
\label{thm:fpt:n}%
The \EQUIJIT\ problem is solvable in $2^{2^{O(n\log n)}} \cdot m^{O(1)}$ time.
\end{theorem}
\begin{proof}

We prove this by providing an ILP formulation of the problem. The program will use $2^{\bigO(n \log n)}$ variables and $2^{\bigO(n \log n)}$ constraints; the proof of the theorem will then directly follow from \Cref{thm:Lenstra}, simply plug in the number of variables to obtain the upper bound $2^{2^{\bigO(n \log n)}}$ on the running time.

Recall that every daily conflict graph $G_i$ is an interval graph with $n$ vertices. The number of distinct interval graphs with $n$ vertices is bounded by $2^{\bigO(n \log n)}$~\cite{bukh2022enumeration,yang2017enumeration}. For every such graph type $G^t$ and every subset of vertices $I \subseteq \{1,\ldots,n\}$ that is an independent set in $G^t$ we create an integer variable $x^{G^t}_{I}$, whose value indicates how many times we schedule the set $I$ on days of type $G^t$.

We proceed to introduce the three sets of constraints, which will guarantee us a feasible and $k$-fair solution. First, we introduce a set (\ref{constraint:xlarger0}) of $2^{\bigO(n \log n)}$ constraints ensuring that all variables are non-negative.
\begin{equation}
\label{constraint:xlarger0}%
\forall G^t, I : x^{G^t}_{I} \geq 0.
\end{equation}
Next, for every $G^t$ we add a constraint to ensure that the total number of independent sets we schedule on days of type $G^t$ is exactly the number of times that $G^t$ appears in the $\q$ days.
\begin{equation}
\label{constraint:numberG}%
\forall G^t : \sum\limits_{I}   x^{G^t}_{I} = |\{ G_i \mid i \in \{1,\ldots,m\} \wedge G_i = G^t\}|.
\end{equation}
Lastly, we introduce another $n$ constraints, one for each client to ensure that its containing subsets are scheduled at least $k$ times.
\begin{equation}
\label{constraint:kfair}%
\forall j \in \{1,\ldots,n\}: \sum\limits_{I:j \in I} \sum\limits_{G^t} x^{G^t}_{I} \geq k.
\end{equation}

$(\Rightarrow)$
Let $\sigma$ be a feasible and $k$-fair schedule. We compute from $\sigma$ an assignment to the variables of the ILP such that no constraint is violated. First, set all variables to 0, then iterate on each of daily schedules in $\sigma$. When the set $I$ is scheduled on day $i$ and the conflict graph on that day be $G^t$, increment the variable $x^{G^t}_{I}$ by one. Clearly (\ref{constraint:xlarger0}) is met as $\sigma$ cannot schedule some job subset a negative number of times. The schedule $\sigma$ is feasible, hence  (\ref{constraint:numberG}) must be upheld, as the scheduled jobs on every day must be independent. Since $\sigma$ is $k$-fair we know that (\ref{constraint:kfair}) is met as every client~$j$ is scheduled in at least $k$ days, which means that the sum of $x^{G^t}_{I}$ must be at least $k$ for every $G^t$ and every $I$ that contains~$j$. 

$(\Leftarrow)$
Any assignment to the variables that satisfies all constraints can be mapped to a feasible and $k$-fair schedule $\sigma$. It is only required that each client has $k$ of its jobs scheduled, regardless of which of the $\q$ days this happens. We can therefore define an arbitrary order $\pi$ over the subsets of $V$ and iterate over the days. Let the $i$-th day have the conflict graph $G^t$, and let $x^{G^t}_{I}$ be the first non-zero variable over the arbitrary order $\pi$. The schedule $\sigma$ will then schedule the jobs associated with $I$ on the $i$-th day. We then proceed to decrease $x^{G^t}_{I}$ by one, and repeatedly continue to day~$i+1$. 

Suppose that $\sigma$ is infeasible, then there must exist two clients whose jobs intersect on a certain day, let the day type be $G^t$. This is a contradiction to the fact that the assignment upholds (\ref{constraint:numberG}) as the number of independent sets occurrences on days of type $G^t$ must equal the number of scheduled client-subsets on days of type $G^t$, and as no client subset can be scheduled a negative number of times it follows that only independent sets are scheduled on $\sigma$, and therefore it must be feasible. Suppose that $\sigma$ is not $k$-fair, then some client $j$ has less than $k$ jobs scheduled, but this is as well a contradiction as we have that (\ref{constraint:kfair}) is met by the assignment. It follows that $\sigma$ is $k$-fair.
\end{proof}


\section{Summary and Discussion}
\label{sec:conclusion}%

In this paper, we studied a new problem \EQUIJIT\ that models repetitive interval scheduling on a single machine when fairness amongst the different clients is the main objective. We provided a thorough complexity analysis for the problem, and identified for which values of $k$ and~$m$ the problem is NP-hard or not. We also considered the special cases of day-independent processing times or due dates, as well as several different parameterizations for the problem through the lens of parameterized complexity. 

\subsection{Generalizations}

In \EQUIJIT\ all clients have the same fairness parameter, they all submit a single job on each day, and we have a single machine to process all these jobs on every day. Our results easily carry over to some natural generalizations of \EQUIJIT\ that we discuss below:
\begin{itemize}
\item \emph{Client-dependent fairness parameters:} The case where each client~$j$ has its own fairness parameter~$k_j$ is discussed in \Cref{sec:paraw}, where we show a reduction from this variant to \EQUIJIT\ in \Cref{lem:td}.

\item \emph{Clients do not require service on some days:} Consider the variant of \EQUIJIT\ where any client~$j$ may not have a job to be processed on some of the days~$i \in \{1,\ldots,m\}$. We can reduce this generalized variant to \EQUIJIT\ by adding  $\lceil{\frac{\q}{k}}\rceil$ auxiliary clients to the set of clients. We also extend the number of days by $k(\lceil{\frac{\q}{k}}\rceil-\frac{\q}{k})$ days. We set the jobs of the auxiliary clients so that they are all mutually conflicting on all days. We then need to add jobs to each of the $n$ original clients so that every client has a job on every day. On each of the original $m$ days, if some client $j$ does not have a job on some day $i$, we create a job for this client that is in conflict only with the jobs of auxiliary clients. On the additional~$k(\lceil{\frac{\q}{k}}\rceil-\frac{\q}{k})$ days, we set the jobs of the original $n$ clients to conflict with the jobs of all auxiliary clients. See \Cref{fig:generalcase} for an illustration. It is easy to verify that an instance of the generalized variant has a solution of value $k$ if and only if the constructed instance of \EQUIJIT\ as a solution of the same value.

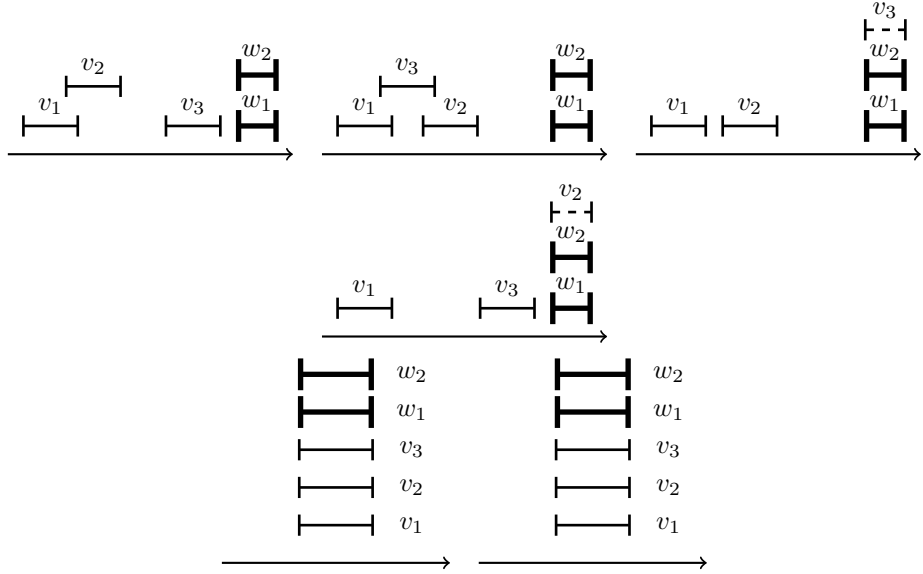
\begin{figure}[H]
\centering

\begin{tikzpicture}[scale=0.75]

  \def\dx{0.05}
  \def\shiftca{4}

    \draw[thick,->] (0,0) -- (5,0) node[anchor=north west] {};

    \draw[line width=1, |-|]
    (0.25,0.5) -- (1.25, .5) node[anchor=north west] {};
    \node at (.75,.85) {$v_{1}$};
    \draw[line width=1, |-|]
    (1.,1.2) -- (2., 1.2) node[anchor=north west] {};
    \node at (1.5,1.55) {$v_{2}$};
    \draw[line width=1, |-|]
    (2.75,0.5) -- (3.75, .5) node[anchor=north west] {};
    \node at (3.25,.85) {$v_{3}$};

    \draw[line width=2, |-|]
    (4,0.5) -- (4.75, .5) node[anchor=north west] {};
    \node at (4.37,.9) {$w_{1}$};
    \draw[line width=2, |-|]
    (4,1.4) -- (4.75, 1.4) node[anchor=north west] {};
    \node at (4.37,1.8) {$w_{2}$};

\end{tikzpicture}
\begin{tikzpicture}[scale=0.75]

  \def\dx{0.05}
  \def\shiftca{4}

    \draw[thick,->] (0,0) -- (5,0) node[anchor=north west] {};

    \draw[line width=1, |-|]
    (0.25,0.5) -- (1.25, .5) node[anchor=north west] {};
    \node at (.75,.85) {$v_{1}$};
    \draw[line width=1, |-|]
    (1.75,0.5) -- (2.75, .5) node[anchor=north west] {};
    \node at (2.35,.85) {$v_{2}$};
    \draw[line width=1, |-|]
    (1.,1.2) -- (2., 1.2)  node[anchor=north west] {};
    \node at (1.5,1.55) {$v_{3}$};

    \draw[line width=2, |-|]
    (4,0.5) -- (4.75, .5) node[anchor=north west] {};
    \node at (4.37,.9) {$w_{1}$};
    \draw[line width=2, |-|]
    (4,1.4) -- (4.75, 1.4) node[anchor=north west] {};
    \node at (4.37,1.8) {$w_{2}$};

\end{tikzpicture}
\begin{tikzpicture}[scale=0.75]

  \def\dx{0.05}
  \def\shiftca{4}
  \draw[thick,->] (0,0) -- (5,0) node[anchor=north west] {};

    \draw[line width=1, |-|]
    (0.25,0.5) -- (1.25, .5) node[anchor=north west] {};
    \node at (.75,.85) {$v_{1}$};
    \draw[line width=1, |-|]
    (1.5,0.5) -- (2.5, .5) node[anchor=north west] {};
    \node at (2.,.85) {$v_{2}$};
    
    \draw[line width=2, |-|]
    (4,0.5) -- (4.75, .5) node[anchor=north west] {};
    \node at (4.37,.9) {$w_{1}$};
    \draw[line width=2, |-|]
    (4,1.4) -- (4.75, 1.4) node[anchor=north west] {};
    \node at (4.37,1.8) {$w_{2}$};
    \draw[line width=1,dashed, |-|]
    (4,2.2) -- (4.75, 2.2) node[anchor=north west] {};
    \node at (4.37,2.55) {$v_{3}$};
\end{tikzpicture}
\begin{tikzpicture}[scale=0.75]

  \def\dx{0.05}
  \def\shiftca{4}
  \draw[thick,->] (0,0) -- (5,0) node[anchor=north west] {};

    \draw[line width=1, |-|]
    (0.25,0.5) -- (1.25, .5) node[anchor=north west] {};
    \node at (.75,.85) {$v_{1}$};
    \draw[line width=1, |-|]
    (2.75,0.5) -- (3.75, .5) node[anchor=north west] {};
    \node at (3.25,.85) {$v_{3}$};
    
    \draw[line width=2, |-|]
    (4,0.5) -- (4.75, .5) node[anchor=north west] {};
    \node at (4.37,.9) {$w_{1}$};
    \draw[line width=2, |-|]
    (4,1.4) -- (4.75, 1.4) node[anchor=north west] {};
    \node at (4.37,1.8) {$w_{2}$};
    \draw[line width=1,dashed, |-|]
    (4,2.2) -- (4.75, 2.2) node[anchor=north west] {};
    \node at (4.37,2.55) {$v_{2}$};

\end{tikzpicture}

\begin{tikzpicture}

    \def\dx{0.05}
  \def\shiftca{4}
  \draw[thick,->] (0,0) -- (3,0) node[anchor=north west] {};

    \draw[line width=1, |-|]
    (1,0.5) -- (2, .5) node[anchor=north west] {};
    \node at (2.5, .5) {$v_{1}$};
    \draw[line width=1, |-|]
    (1,1.) -- (2, 1.) node[anchor=north west] {};
    \node at (2.5, 1.) {$v_{2}$};
    \draw[line width=1, |-|]
    (1,1.5) -- (2, 1.5) node[anchor=north west] {};
    \node at (2.5, 1.5) {$v_{3}$};

    \draw[line width=2, |-|]
    (1,2.) -- (2, 2.) node[anchor=north west] {};
    \node at (2.5, 2.) {$w_{1}$};
    \draw[line width=2, |-|]
    (1,2.5) -- (2, 2.5) node[anchor=north west] {};
    \node at (2.5, 2.5) {$w_{2}$};
\end{tikzpicture}
\begin{tikzpicture}

  \def\dx{0.05}
  \def\shiftca{4}
  \draw[thick,->] (0,0) -- (3,0) node[anchor=north west] {};

    \draw[line width=1, |-|]
    (1,0.5) -- (2, .5) node[anchor=north west] {};
    \node at (2.5, .5) {$v_{1}$};
    \draw[line width=1, |-|]
    (1,1.) -- (2, 1.) node[anchor=north west] {};
    \node at (2.5, 1.) {$v_{2}$};
    \draw[line width=1, |-|]
    (1,1.5) -- (2, 1.5) node[anchor=north west] {};
    \node at (2.5, 1.5) {$v_{3}$};

    \draw[line width=2, |-|]
    (1,2.) -- (2, 2.) node[anchor=north west] {};
    \node at (2.5, 2.) {$w_{1}$};
    \draw[line width=2, |-|]
    (1,2.5) -- (2, 2.5) node[anchor=north west] {};
    \node at (2.5, 2.5) {$w_{2}$};
    
\end{tikzpicture}
\caption{The construction of \EQUIJIT\ from the general case by which $\q=4$ and $k=3$.
The reduction adds $3(\lceil{\frac{4}{3}}\rceil-\frac{4}{3})=2$ additional days and $\lceil{\frac{4}{3}}\rceil=2$ auxiliary clients $w_1$ and~$w_2$. The top figure depicts the original days, while the bottom figure depicts the additional days. The jobs in \textbf{bold} are jobs belonging to the auxiliary clients, and the dashed jobs are new jobs that are added for the original clients and on the original days.}
\label{fig:generalcase}
\end{figure}

\item \emph{Agreeable due dates:} An instance of \EQUIJIT\ is said to have agreeable due dates if $d_{i,j} \leq d_{i,j+1}$ for all $i \in \{1,\ldots,m\}$ and $j \in \{1,\ldots,n-1\}$. Note that the case of day-independent due dates is a special case of agreeable due dates. Agreeable conditions are widely used in the scheduling literature to derive polynomial time algorithms for NP-hard problems~\cite{bampis2012speed,gilenson2021scheduling,ma2013online,liu2007scheduling,shabtay2022state}. We can create from every \textit{agreeable due dates} instance an equivalent instance by which the due dates are \textit{day-independent}. To do so, we simply set $d_{i,j}=j$ to the $j$'th job in the ordering. Then, we set $p_{i,j} = \ell$ where $\ell$ is the smallest job index which is in conflict with job $j$ on day $i$. It is not difficult to see that we obtain an equivalent instance of \OPEQUIJIT\ in this way. 

\item \emph{Multiple Machines:} Most of our algorithms also extend naturally to the case where there are~$M$ identical machines on each day to process the clients' jobs: The running time of the algorithm in \Cref{thm:fpt:n} increases to $2^{2^{O(Mn\log(Mn))}} \cdot m^{O(1)}$ by replacing variables $x^{G^t}_{I}$ with variables of the form~$x^{G^t}_{I_1,\ldots,I_M}$. We get an algorithm running time 
$M^{\tau m}\cdot n$
by extending the algorithm in \Cref{thm:fpt:w+m} so that it considers the schedules on $M$ machines on each day for each node of the tree decomposition. For the case of day-independent processing times and due dates, we can easily reduce to \SDEQUIJIT\ with $Mm$ days, and the algorithms of \Cref{thm:bipartite} and Theorem~\ref{thm:opxp} can be extended to algorithms with running times of~$O(\sqrt{M}n^{1.5}m^{2.5}+(nmM)^{1.5})$ and~$O((mM)^{k+1}n^{mM+1})$, respectively. A notable exception is \Cref{thm:km-1} and the case of $m=k-1$ which does not naturally extend to the multiple machine case.  
\end{itemize} 

\subsection{Open problems}

Despite our thorough investigation into the complexity of \EQUIJIT\, there are still several interesting open questions left by our work. Below we give only a partial list: 
\begin{itemize}
\item Can \UEQUIJIT\ be solved in $O(mn)$ time?
\item Does \OPEQUIJIT\ admit an algorithm running in~$n^{o(m)}$ time? Note that \Cref{thm:k2m3} excludes an FPT algorithm for $m$.
\item Can the $M$-machine case be solved in polynomial-time for~$k=m-1$?
\end{itemize}
\noindent Naturally, apart from the Just-In-Time objective $\sum_i Z_{i,j}$, there are several natural other objectives to explore in the context of fair repetitive scheduling, and this paper should serve as an encouragement to explore this interesting new area.

\backmatter

\bmhead{Supplementary information}
Not applicable.




\bmhead{Acknowledgements}
Not applicable.



\section*{Declarations}

Klaus Heeger, Danny Hermelin, Yuval Itzhaki, and Dvir Shabtay are supported by the ISF, grant No.~1070/20. Hendrik Molter is supported by the ISF, grants No.~1070/20 and No.~1456/18, and European Research Council, grant number 949707. 

The authors have no relevant financial or non-financial interests to disclose.




\bibliography{sn-bibliography}%

\end{document}